\documentclass[12pt,onecolumn,accepted=2019-07-02]{quantumarticle}
\pdfoutput=1 
\usepackage[T1]{fontenc}

\usepackage{tikz}
\usetikzlibrary{positioning, calc}
\usetikzlibrary{calc}
\usetikzlibrary{arrows}
\usepackage{tikz-3dplot}
\usepackage{graphicx}
\usepackage{mdwlist}
\usepackage{color}
\usepackage{array}
\usepackage{doi}
\usepackage{url,hyperref}
\usepackage[numbers]{natbib}
\usetikzlibrary{fadings}
\usetikzlibrary{decorations.pathmorphing}
\usepackage{mathrsfs}

\tikzset{snake it/.style={decorate, decoration=snake}}

\usetikzlibrary{decorations.pathreplacing,decorations.markings}

\tikzset{
    >=stealth',
    punkt/.style={
           rectangle,
           rounded corners,
           draw=black, very thick,
           text width=6.5em,
           minimum height=2em,
           text centered},
    pil/.style={
           ->,
           thick,
           shorten <=2pt,
           shorten >=2pt,},
  on each segment/.style={
    decorate,
    decoration={
      show path construction,
      moveto code={},
      lineto code={
        \path [#1]
        (\tikzinputsegmentfirst) -- (\tikzinputsegmentlast);
      },
      curveto code={
        \path [#1] (\tikzinputsegmentfirst)
        .. controls
        (\tikzinputsegmentsupporta) and (\tikzinputsegmentsupportb)
        ..
        (\tikzinputsegmentlast);
      },
      closepath code={
        \path [#1]
        (\tikzinputsegmentfirst) -- (\tikzinputsegmentlast);
      },
    },
  },
  mid arrow/.style={postaction={decorate,decoration={
        markings,
        mark=at position .5 with {\arrow[#1]{stealth'}}
      }}}
}

\usepackage{hyperref}
\usepackage{slashed}

\usepackage{float}		
\usepackage{graphicx}		
\usepackage{subcaption}
\usepackage{amsmath}

 \newcommand{\ket}[1]{|#1\rangle}
\newcommand{\ketbra}[2]{|#1\rangle\!\langle#2|}

\newtheorem{theorem}{Theorem}

\newtheorem{definition}[theorem]{Definition}

\newtheorem{protocol}[theorem]{Protocol}
\newenvironment{proof}[1][Proof]{\noindent\textbf{#1.}}{\ \rule{0.5em}{0.5em}}

\renewcommand{\theenumi}{\roman{enumi}}

\begin{document} 

\title{Localizing and excluding quantum information; or, how to share a quantum secret in spacetime}

\author[a]{Patrick Hayden}
\address{Stanford University}
\ead{phayden@stanford.edu}

\author[b]{Alex May}
\address{The University of British Columbia}
\ead{may@phas.ubc.ca}

\begin{abstract}
When can quantum information be localized to each of a collection of spacetime regions, while also excluded from another collection of regions? We answer this question by defining and analyzing the localize-exclude task, in which a quantum system must be localized to a collection of authorized regions while also being excluded from a set of unauthorized regions. This task is a spacetime analogue of quantum secret sharing, with authorized and unauthorized regions replacing authorized and unauthorized sets of parties. Our analysis yields the first quantum secret sharing scheme for arbitrary access structures for which the number of qubits required scales polynomially with the number of authorized sets. We also study a second related task called state-assembly, in which shares of a quantum system are requested at sets of spacetime points. We fully characterize the conditions under which both the localize-exclude and state-assembly tasks can be achieved, and give explicit protocols. Finally, we propose a cryptographic application of these tasks which we call party-independent transfer.
\end{abstract}

\maketitle

\pagebreak

\tableofcontents

\section{Introduction}

The study of the interplay between quantum theory and relativity has recently begun a new chapter with the consideration of quantum information tasks in a Minkowski space background~\cite{fuentes2005alice,rideout2012fundamental,martin2013processing}. For instance, the study of information causality \cite{pawlowski2009information} and of causal operators \cite{beckman2001causal} has given further insight into ties between information processing and relativity. Along with other results in this area \cite{hayden2016summoning,adlam2015quantum,hayden2016spacetime,kent2018unconstrained}, these can be placed into the general framework of quantum tasks in Minkowski space \cite{kent2012quantum}.

One task of particular interest is \emph{summoning}, defined by  Kent \cite{kent2013no}, where the associated no-summoning theorem is a statement of no-cloning appropriate to the spacetime setting. We have also argued that a generalization of the summoning task \cite{hayden2016summoning} provides an operational framework within which to study how quantum information can move through spacetime. The importance of having such a framework is highlighted by recent subtle questions concerning spacetime structure and the no-cloning principle in the context of black holes \cite{hayden2007black,almheiri2013black}. Understanding how a quantum system may be delocalized in Minkowski space should be a useful step towards understanding such fundamental puzzles. 

The study of quantum tasks in Minkowski space has been given a second motivation with the discovery of cryptographic protocols that exploit the properties of either or both of quantum mechanics and special relativity. Bit-commitment is a well-known example \cite{kent12,flyingqudits11}; other examples include coin flipping \cite{kent99}, key distribution (where signalling constraints enter into some security proofs \cite{barrett05,barrett12}), and two spacetime analogues of oblivious transfer dubbed location-oblivious transfer \cite{location11} and spacetime-constrained oblivious transfer \cite{pitalua2016spacetime}.

In quantum secret sharing, a central result of quantum cryptography, a quantum system is distributed among many parties such that only certain subsets of parties may collectively use their shares to reconstruct the system. Other subsets of parties are required to not be able to learn any information about the secret from their shares. In the context of quantum tasks in Minkowski space, where the movement of information in spacetime is central, and in the context of relativistic quantum cryptography, it is natural to consider a spacetime generalization of quantum secret sharing. 

To do this we replace the notions of authorized and unauthorized sets of parties with authorized and unauthorized spacetime regions. We define the localize-exclude task, where the goal is to move a quantum system through spacetime in such a way that it is localized to each of the authorized regions and excluded from the unauthorized ones. Figure \ref{fig:firstexample} gives a simple example. In theorem \ref{thm:localexclude}, we find necessary and sufficient conditions for completing the localize-exclude task. To argue that the localize-exclude task is a natural spacetime generalization of secret sharing, we show in the main text that there is a simple construction that embeds any quantum secret sharing scheme as a localize-exclude task, and that the conditions of this theorem reduce to those for quantum secret sharing in that case. 

\begin{figure}
\centering
\begin{tikzpicture}[scale=0.45]

\coordinate (P) at (0cm, 0cm);

\draw[dashed, blue] (-5,5) circle (2cm);
\node at (-5,5) {$\mathscr{A}_1$};
\draw[dashed,blue] (5,15) circle (2cm);
\node at (5,15) {$\mathscr{A}_2$};

\node[below] at (0,0) {s};
\draw[fill=yellow] (0,0) circle (0.15cm);

\draw[dashed,red] (4,5) -- (6,7) -- (-4,17) -- (-6,15) -- (4,5);
\node at (0,11) {$\mathscr{U}_1$};

\begin{scope}[shift={($(-5,0)$)}]       
    \draw[->] (-3,0) -- (-2,0);
    \node [right] at (-2,0) {$x$};
    \draw[->] (-3,0) -- (-3,1);
    \node [left] at (-3,1) {$t$};
\end{scope}

\end{tikzpicture}
\caption{An example of a localize-exclude task. A single copy of an unknown quantum system is initially localized near the spacetime point $s$, and needs to be localized to within regions $\mathscr{A}_1$ and $\mathscr{A}_2$, while avoiding region $\mathscr{U}_1$. Theorem \ref{thm:localexclude} shows that this is possible to do.}
\label{fig:firstexample}
\end{figure}
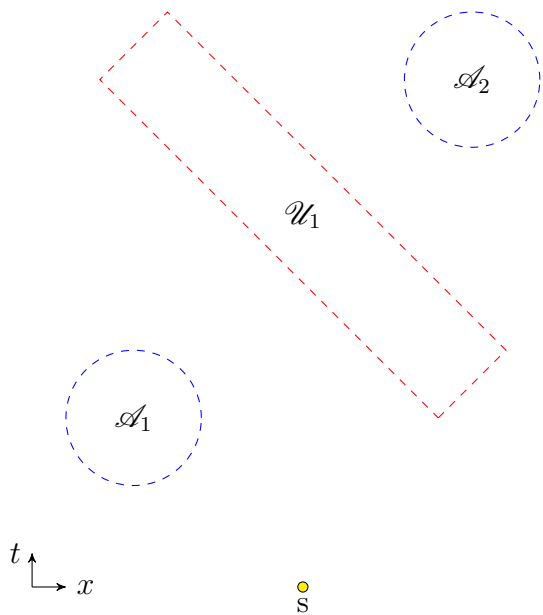

In the summoning task one party, Bob, puts in requests for the quantum system at certain spacetime points, asking that the system be returned at one of another set of points. The localize-exclude task removes this structure, but adds a notion of unauthorized region. It is interesting to also consider a task in the request-return setting, but which includes unauthorized regions. In this \emph{state-assembly} task, we consider many parties Bob$_i$ who may each request a share of the quantum system at an associated spacetime region $D_i$. Alice should respond to the collection of requests given by the Bobs in a careful way: she should hand over a collection of shares sufficient to construct a single copy of the system when the collection of requests is authorized, and she should not reveal any information about the system when that collection is unauthorized. The conditions for Alice to complete this task are the same as for localize-exclude in the case of causally separated regions, but differ when non-trivial causal structures are considered. In theorem \ref{thm:authorizedunauthorizedassembly} below we precisely characterize the conditions under which this task can be completed, and describe an explicit protocol for completing it when it is possible.

Together the state-assembly and localize-exclude tasks provide a rich set of scenarios to consider. We suggest \emph{party-independent transfer} as a potential cryptographic application of this framework, a task where two other parties wish to receive information from Alice and want the information they receive to be both private and independent of their identity. We propose a protocol for completing this task which is built on the state-assembly task. Establishing the security of this protocol we leave to future work.

The layout of this paper is as follows. Section \ref{sec:localexclude} gives the necessary definitions to study localization to arbitrary spacetime regions and proves theorem \ref{thm:localexclude}, which characterizes the localize-exclude task. We discuss the relation between localize-exclude and quantum secret sharing in the same section. In section \ref{sec:assembly} we discuss state-assembly and give its characterization. In section \ref{sec:partyindependenttransfer} we study the party-independent transfer task. Two appendices are included which clarify the relationship of this work to earlier work on summoning. The first shows that state-assembly is equivalent to a certain summoning task, and the second addresses the points raised by Adlam and Kent \cite{adlam2015quantum} against interpreting summoning tasks in terms of the localization of information.

\section{Localizing and excluding quantum information}\label{sec:localexclude}

\subsection{Localizing quantum information to many regions}

As a first step towards characterizing the localize-exclude task we discuss the problem of localizing quantum information to a collection of spacetime regions, leaving excluded regions to the next section. To do this we consider the following setting. Alice holds the $A$ subsystem of a pure state $\ket{\Psi}_{RA}$, with $A$ recorded into a collection of classical and quantum systems held within secure laboratories not accessible to her adversary, Bob\footnote{Alice and Bob are both agencies, who have many agents that may be distributed to many different laboratories.}. We would like to ask where system $A$ is. For instance, Alice might have recorded $A$ into an error-correcting code and distributed the shares of this code to various laboratories. Further, she might be constantly rerouting these shares between labs, so that shares are held only at certain labs between specified times. 

We can ask where the subsystem is in spacetime by temporarily relaxing the security of Alice's labs --- we give Bob access to some collection of Alice's labs for certain time intervals. If by accessing these labs Bob is able to prepare the $A$ system (potentially making use of later data processing), we say that system $A$ was localized to the collection of labs and intervals of time Bob accessed. More generally, we can abstract away from the language of labs and time intervals and give a more general definition. 
\begin{definition}\label{def:localized}
Suppose one party, Alice, holds system $A$ of a quantum state $\ket{\Psi}_{AR}$. Then we say the subsystem $A$ is \textbf{localized} to a spacetime region $\Sigma$ if a second party, Bob, for whom the state is initially unknown is able to prepare the $A$ system by collecting quantum and classical systems from within $\Sigma$, and then applying later data processing.
\end{definition}
Conversely, if Bob is unable to learn anything about $A$ we say the system is \textbf{\emph{excluded}} from $\Sigma$. Note that the later data processing referred to in the definition may occur outside of the region $\Sigma$. Further, a system may be neither localized nor excluded from a region if partial information about the system is available there. 

To be more precise we should specify how it is verified that Bob holds the $A$ system after he has accessed $\Sigma$. One natural possibility is to introduce a third party, call him Charlie, who plays the role of a referee. We have Charlie hold both the purifying system $R$ of $\ket{\Psi}_{AR}$ as well as a classical description $\ket{\Psi}_{AR}$. To verify Bob holds the system then, we have Bob pass the $A$ system to Charlie, who performs a projective measurement of the $AR$ system in a basis that includes $\ket{\Psi}_{AR}$. If Alice can pass Charlie's test with certainty, we declare that Alice localized the system to $\Sigma$. Of course, Alice will also pass this test with some probability so long as Charlie's final state has non-trivial overlap with $\ket{\Psi}_{AR}$\footnote{In the context of the localize-exclude tasks we consider later, we may be interested in whether or not a quantum system can be localized to two or more regions with fixed relative positions, rather than if a system can be localized to one particular spacetime region. In this case, and when the spacetime is suitably translation invariant, one can consider repeating the task many times (sequentially or in parallel). In this scenario Charlie could determine with what probability Alice is able to complete the task.}.

It is interesting to compare this notion of localizing a quantum system to a spacetime region to a notion of spacetime localization based on the summoning task \cite{hayden2016summoning}. Perhaps the key distinction is that, in the definition given here, information processing may occur outside the spacetime region in order to prepare the system. This point carries with it certain subtleties that are taken up in appendix \ref{appendix:equivalence} and the discussion. The key advantage of definition \ref{def:localized} however is its applicability to regions of arbitrary shape. 

One strategy for hiding a quantum system from Bob would be for Alice to send $A$ into a region $\Sigma$, while also sending various decoys $A_{d1},A_{d2},...$ so that Bob, though he may collect all of the systems $A, A_{d1}, A_{d2},...$ is left unsure as to which system to hand to Charlie. This reveals a finer point to definition \ref{def:localized}: the system Bob is searching for may enter $\Sigma$, but if appropriate classical instructions do not also enter $\Sigma$ (in this case a label denoting which system actually holds $A$), then definition \ref{def:localized} says the system is not localized there. To avoid confusion around this point we will always have Alice, at some early time, reveal the classical instructions that constitute her protocol to Bob. The only information Alice will not broadcast is a classical string $k$ (as well as the quantum system itself). As we will see, protocols where Alice holds only a secret key $k$ and reveals all other details to Bob are sufficient to complete any physically possible localize-exclude task, so this restriction on Alice amounts to a useful simplification of notation and language.

In the protocols we construct Alice will encode her quantum system into an error-correcting code that corrects erasure errors, and then apply a quantum one-time pad to each of the shares in the quantum code. Alice does not broadcast the classical strings used in the one-time pads; taken together these constitute her secret key $k$. However, she does reveal her procedure for putting $A$ into an error-correcting code and applying the one-time pad, and reveals the spacetime trajectories of each share in the code. Within this context, Bob reconstructs $A$ by accessing a region $\Sigma$ whenever a correctable subset of shares in the error-correcting code along with their corresponding classical keys from the the one-time pad pass through $\Sigma$. 

Definition \ref{def:localized} specifies what is meant by a quantum system being localized to a single spacetime region. To extend this to multiple regions, we define the localize task as follows.
\begin{definition} A \textbf{localize task} is a task involving two agencies, Alice and Bob, specified by a tuple $\{A,s,\{\mathscr{A}_1,...,\mathscr{A}_n\}\}$, consisting of:
\begin{itemize}
    \item A quantum system $A$. In general $A$ may be a subsystem of some overall pure state $\ket{\Psi}_{AR}$. The state on $AR$ is unknown to both Alice and Bob. 
    \item A start point $s$, at which Alice initially holds system $A$
    \item A collection of spacetime regions $\{\mathscr{A}_1,...,\mathscr{A}_n\}$, which we call the authorized regions
\end{itemize}
Alice successfully completes the task if Bob is able to prepare system $A$ after he accesses any one of the $\mathscr{A}_i$.
\end{definition}
If Alice is able to successfully complete the localize task with regions $\{\mathscr{A}_1,...,\mathscr{A}_n\}$ we say she has localized the system to each of those regions. The authorized regions may be of arbitrary shape and may overlap. 

\begin{figure}
\begin{center}
\begin{subfigure}{.45\textwidth}
\begin{center}
\begin{tikzpicture}[scale=0.3]

\draw (-4,4) circle (3);
\node at (-4,4) {$\Sigma_i$};
\draw (4,14) circle (3);
\node at (4,14) {$\Sigma_j$};

\draw[thick,postaction={on each segment={mid arrow}}] (-2,4) -- (4,12);

\draw plot [mark=*, mark size=4] coordinates{(-2,4)};
\node[below] at (-2,4) {$q_i$};
\draw plot [mark=*, mark size=4] coordinates{(4,12)};
\node[right] at (4,12) {$q_j$};

\end{tikzpicture}
\end{center}
\caption{}
\label{fig:CausallyConnected}
\end{subfigure}
\hfill
\begin{subfigure}{.45\textwidth}
\begin{center}
\begin{tikzpicture}[scale=0.35]
  
\draw[lightgray,fill=lightgray] (-6,8) -- (0,14) -- (6,8) -- (6,6) -- (0,0) -- (-6,6) -- (-6,8);
\draw[thick, black,fill=gray] (-6,6) -- (-6,8) -- (6,8) -- (6,6) -- (-6,6);
\node at (0,7) {$\Sigma$};

\draw plot [mark=*, mark size=4] coordinates{(2,3)};
\node[below] at (2,3) {$p$};

\draw[thick,postaction={on each segment={mid arrow}}] (2,3) to [out=90,in=-120] (5,10);

\end{tikzpicture}
\end{center}
\caption{}
\label{fig:DomainOfDependence}
\end{subfigure}
\caption{Two geometric notions used in the text. a) Two causally connected regions. Two spacetime regions $\Sigma_i$ and $\Sigma_j$ are said to be causally connected if there is a point $q_i$ in $\Sigma_i$ and $q_j$ in $\Sigma_j$ such that there is a causal curve from $q_i$ to $q_j$, or from $q_j$ to $q_i$. b) The domain of dependence (light grey) of a spacetime region $\Sigma$ (dark grey). The domain of dependence is defined as the set of all points $p$ in the spacetime such that all causal curves passing through $p$ must also enter $\Sigma$.}
\label{fig:geometryDefs}
\end{center}
\end{figure}
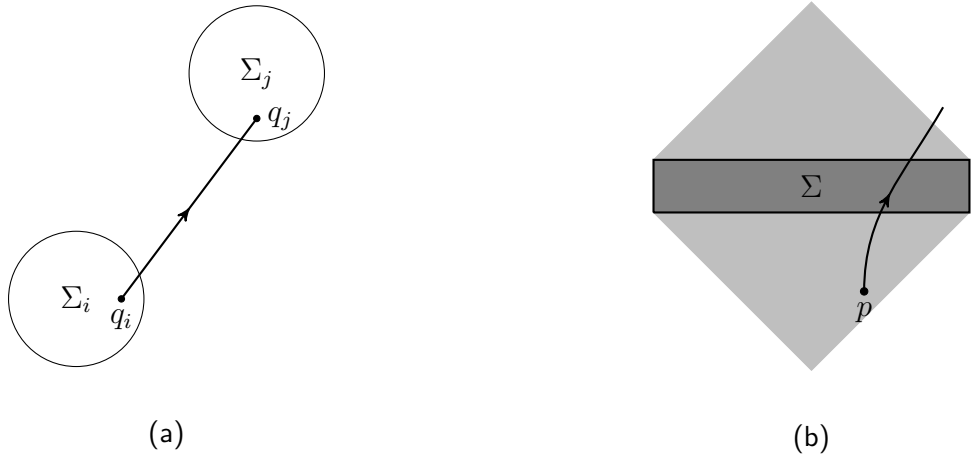

To analyze this task it is useful to introduce some language. We give the following definition which specifies a relation between pairs of spacetime regions. 
\begin{definition}\label{def:connected}
Two spacetime regions $\Sigma_i$ and $\Sigma_j$ are said to be \textbf{causally connected} if there is a point $q_i$ in $\Sigma_i$ and $q_j$ in $\Sigma_j$ such that there is a causal curve from $q_i$ to $q_j$, or from $q_j$ to $q_i$.
\end{definition}
We illustrate this definition in figure \ref{fig:CausallyConnected}. If two regions are not causally connected we say they are \textbf{causally disjoint}. In the context of the localize-exclude task discussed in the next section we will also need one further definition relating to spacetime geometry.
\begin{definition}
The \textbf{domain of dependence} of a spacetime region $\Sigma$, denoted $D(\Sigma)$, is the set of all points $p$ such that every causal curve through $p$ must also enter $\Sigma$. 
\end{definition}
This definition is illustrated in figure \ref{fig:DomainOfDependence}. 

\begin{figure}
\begin{center}
\begin{tikzpicture}[scale=0.6]

\draw[dashed, blue] (-5,8) circle (1.5);
\node at (-5,8) {$\mathscr{A}_1$};
\draw[dashed,blue] (5,8) circle (1.5);
\node at (5,8) {$\mathscr{A}_2$};

\draw[->] (0,0) -- (0,1);

\node[below] at (0,0) {s};
\draw[fill=yellow] (0,0) circle (0.15cm);

\draw[black] (-1,1) -- (1,1) -- (1,2) -- (-1,2) -- (-1,1);
\draw (-0.6,1.1) arc (180:0:0.6);
\draw (0,1.1) -- (0.7,1.8);

\begin{scope}[shift={($(-5,0)$)}]       
    \draw[->] (-3,0) -- (-1,0);
    \node [below] at (-1,0) {$x$};
    \draw[->] (-3,0) -- (-3,2);
    \node [left] at (-3,2) {$t$};
\end{scope}

\draw[double,postaction={on each segment={mid arrow}}] (0,2) -- (-4,7);
\draw[double,postaction={on each segment={mid arrow}}] (0,2) -- (4,7);

\draw[dashed,blue] (11,0) circle (1.5);
\node at (11,0) {$\mathscr{A}_1$};
\draw[fill=black] (10,0) circle (0.15cm);
\node[left] at (10,0) {$p_1$};
\draw[dashed,blue] (11,5) circle (1.5);
\node at (11,5) {$\mathscr{A}_2$};
\draw[fill=black] (10,5) circle (0.15cm);
\node[left] at (10,5) {$p_2$};

\draw[thick] (9.5,-2) to [out=80,in=-100] (10,0);
\draw[thick,postaction={on each segment={mid arrow}}] (10,0) to [out=80,in=-80] (10,5);
\draw[thick] (10,5) to [out=100,in = -80] (9.5,7);
\node[below] at (9.5,-2) {$\bar{E}$};

\draw[thick,postaction={on each segment={mid arrow}}] (1,-2) to [out=110,in=-90] (0.5,1);
\node[below] at (1,-2) {$E$};

\draw[dashed] (1.5,-2.5) -- (9.15,-2.5);
\node[below] at (5.325,-2.5) {$\ket{\Psi^+}$};

\end{tikzpicture}
\end{center}
\caption{An arrangement of two authorized regions that has the minimal requirements to satisfy the conditions of theorem \ref{thm:simplelocalize}. By the first condition $\mathscr{A}_1$ and $\mathscr{A}_2$ are causally connected. This guarantees the existence of a point $p_1$ in $\mathscr{A}_1$ which is in the causal future of some point $p_2$ in $\mathscr{A}_2$ (up to relabelling). The second condition gives that each region have at least one point in the future light cone of $s$. However, the regions $\mathscr{A}_1$ and $\mathscr{A}_2$ may be disconnected (as shown here) and so satisfy this requirement while having the points $p_1,p_2$ be outside the future light cone of $s$. To localize a system $A$ to both regions a maximally entangled state $\ket{\Psi^+}_{E\bar{E}}$ is shared between $s$ and $p_1$. Near to $s$ the $A$ system is teleported using this entanglement, and the entangled system at $p_1$ is sent to $p_2$. Meanwhile, the classical measurement outcomes from the teleportation protocol are sent to the points in $\mathscr{A}_1$ and $\mathscr{A}_2$ which are in the causal future of $s$. Each region has both the classical measurement outcomes and the entangled particle pass through it, so the $A$ system is localized to each.}
\label{fig:generaltworegion}
\end{figure}
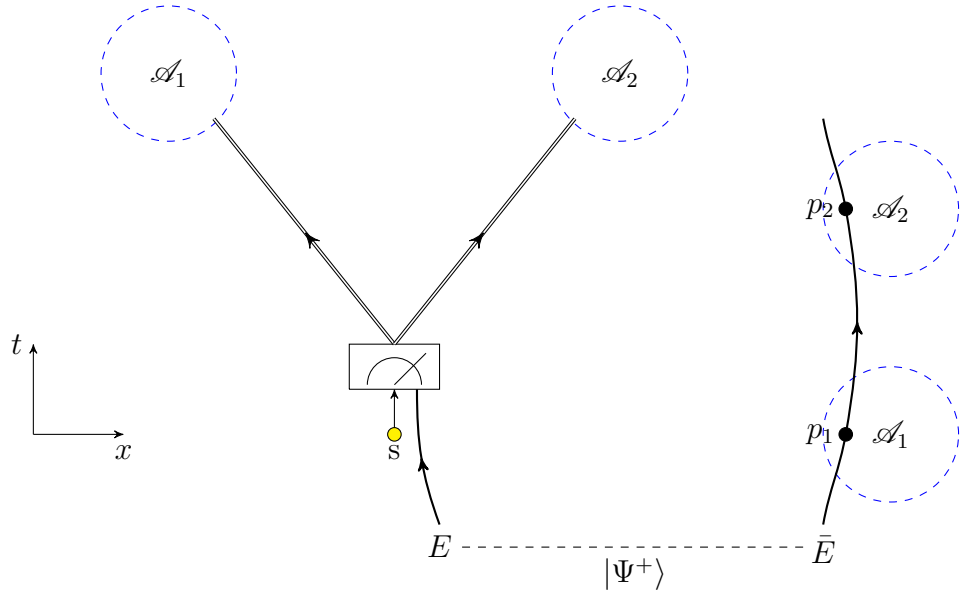

As a first step towards the more general scenario consider the localization of a quantum system to two authorized regions $\mathscr{A}_1$ and $\mathscr{A}_2$. 
\begin{theorem}\label{thm:simplelocalize}
Given a quantum system initially localized near a spacetime point $s$, the system may be localized to both of the spacetime regions $\mathscr{A}_1$ and $\mathscr{A}_2$ if and only if the following two conditions hold.
\begin{enumerate}
    \item $\mathscr{A}_1$ and $\mathscr{A}_2$ both have a point in the future light cone of $s$.
    \item $\mathscr{A}_1$ and $\mathscr{A}_2$ are causally connected. 
\end{enumerate}
\end{theorem}
\begin{proof} \,First, note that if an authorized region is entirely outside the future light cone of the start point then successfully localizing the system to that region would constitute superluminal communication. Thus, the first condition is necessary. To see necessity of the second condition suppose there exists a protocol for localizing a quantum system to two causally disjoint regions $\mathscr{A}_1$ and $\mathscr{A}_2$. Then by definition it is possible to construct the system by accessing the region $\mathscr{A}_1$, and by accessing $\mathscr{A}_2$. By causality however accessing region $\mathscr{A}_1$ cannot affect the system constructed from $\mathscr{A}_2$, and vice versa, so it would be possible to construct two copies of the quantum system. But this constitutes cloning, so no such protocol can exist.

To understand sufficiency we construct a task with the minimal properties specified by the two assumed conditions. Such a task is shown in figure \ref{fig:generaltworegion}. There, a point $p_1\in \mathscr{A}_1$ is causally connected to $p_2\in \mathscr{A}_2$, and each of $\mathscr{A}_1$ and $\mathscr{A}_2$ have a point in the future light cone of $s$. However, $p_1$ and $p_2$ sit outside the future light cone of $s$. Nonetheless it is straightforward to complete such a task. To do so a system $E$ is maximally entangled with $\bar{E}$, then $E$ is brought to $s$ while $\bar{E}$ is brought to $p_1$. At $s$, $E$ is used to teleport the $A$ system onto the $\bar{E}$ system. The measurement outcome from the teleportation is sent to $\mathscr{A}_1$ and $\mathscr{A}_2$ from $s$. Meanwhile, $\bar{E}$ is sent from $p_1$ to $p_2$. Each authorized region contains the classical measurement outcome and the system $\bar{E}$, so accessing either region allows reconstruction of $A$. 
\end{proof}

We can now move on to understanding localize tasks with arbitrary numbers of authorized regions. We find in particular that it is only the structure of causal connections between pairs of regions and the start point that are needed to characterize a task as possible or impossible.
\begin{theorem}\label{thm:localize}
Given a quantum system $A$ initially localized near a spacetime point $s$, the system may be localized to each spacetime region $\mathscr{A}_i$ in a collection $\{\mathscr{A}_1,....,\mathscr{A}_n\}$ if and only if the following two conditions hold.
\begin{enumerate}
\item Each region $\mathscr{A}_i$ has at least one point in the causal future of $s$
\item Each pair of regions $(\mathscr{A}_i,\mathscr{A}_j)$ is causally connected. 
\end{enumerate}
\end{theorem}
\begin{proof}\,Necessity of the two conditions follows from the same arguments as in the two region case given as theorem \ref{thm:simplelocalize}: localizing a system to a region outside of its future light cone violates no signaling, and localizing a system to two spacelike separated regions would allow two copies of the system to be produced. 

To demonstrate sufficiency we construct an explicit protocol for completing any task satisfying the two conditions. To this end it is useful to introduce a directed graph $G$ which describes the causal structure of the task: for each authorized region $\mathscr{A}_i$ introduce a vertex, also labelled $\mathscr{A}_i$, to the graph. For each pair of regions $(\mathscr{A}_i,\mathscr{A}_j)$ such that there is a point in $\mathscr{A}_j$ connected by a causal curve to a point in $\mathscr{A}_i$ introduce a directed edge $(\mathscr{A}_i\rightarrow \mathscr{A}_j)$. An example of a task and its associated graph is given as figure \ref{fig:disconnectedtriangle}.

From the no-cloning theorem it follows that some quantum information must be shared between every pair of authorized regions. In our construction these quantum systems that move between pairs of authorized regions form the shares of an error-correcting code. In particular, for each edge in the graph $G$ we associate one share. In theorem \ref{thm:simplelocalize} and figure \ref{fig:generaltworegion} we showed how to localize a quantum system to two authorized regions whenever they share a causal connection. We can execute this protocol on the shares of our error-correcting code to ensure the share associated to edge $\mathscr{A}_i\rightarrow \mathscr{A}_j$ is localized to both $\mathscr{A}_i$ and $\mathscr{A}_j$. To complete the task then, our error-correcting code should have the property that, given any vertex, the set of shares associated to the edges attached to that vertex are sufficient to construct the initial system $A$. We illustrate the requirement on this code in figure \ref{fig:graphcode}. 

In fact, given that every pair of vertices in this graph share an edge, which is guaranteed by condition (ii), such error-correcting codes have already been constructed. To encode finite-dimensional quantum systems we constructed such codes using the codeword-stabilized formalism in the context of a similar summoning problem \cite{hayden2016summoning}. Constructions for continuous variable systems have also been given \cite{hayden2016spacetime} and then adapted to the finite-dimensional case \cite{wu2017efficient}. In the code-word stabilized construction a single logical qubit is recorded using 2 physical qubits for each edge in the graph, resulting in a total of $2 {n \choose 2}$ physical qubits for $n$ the number of authorized regions. 
\end{proof}

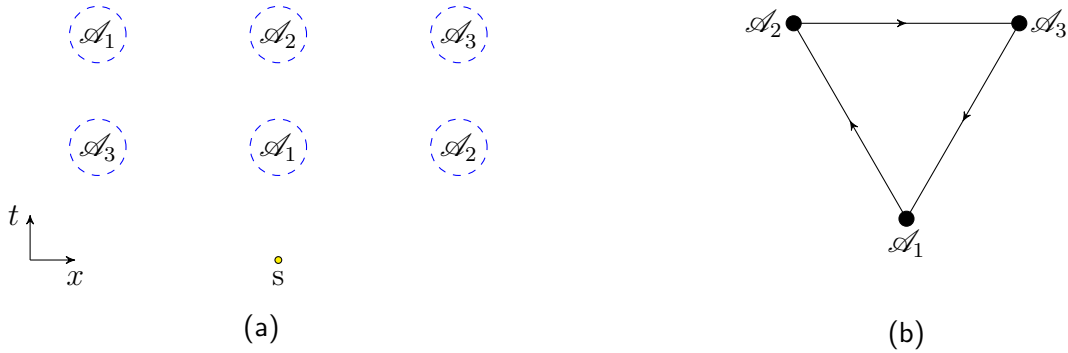
\begin{figure}
\begin{center}
\begin{subfigure}{.45\textwidth}
\begin{tikzpicture}[scale=0.3]

\draw[dashed, blue] (0,0) circle (1.25);
\node at (0,0) {$\mathscr{A}_3$};
\draw[dashed,blue] (0,5) circle (1.25);
\node at (0,5) {$\mathscr{A}_1$};

\draw[dashed, blue] (8,0) circle (1.25);
\node at (8,0) {$\mathscr{A}_1$};
\draw[dashed,blue] (8,5) circle (1.25);
\node at (8,5) {$\mathscr{A}_2$};

\draw[dashed, blue] (16,0) circle (1.25);
\node at (16,0) {$\mathscr{A}_2$};
\draw[dashed,blue] (16,5) circle (1.25);
\node at (16,5) {$\mathscr{A}_3$};

\begin{scope}[shift={($(0,-5)$)}]       
    \draw[->] (-3,0) -- (-1,0);
    \node [below] at (-1,0) {$x$};
    \draw[->] (-3,0) -- (-3,2);
    \node [left] at (-3,2) {$t$};
\end{scope}

\node[below] at (8,-5) {s};
\draw[fill=yellow] (8,-5) circle (0.15cm);

\end{tikzpicture}
\caption{}
\end{subfigure}
\hfill
\begin{subfigure}{.45\textwidth}
\begin{center}
\begin{tikzpicture}[scale=0.5]
  
\draw[fill=black] (0,0) circle (0.2cm);
\node[below] at (0,0) {$\mathscr{A}_1$};

\draw[postaction={on each segment={mid arrow}}] (0,0) -- (-3,5.2);

\draw[fill=black] (-3,5.2) circle (0.2cm);
\node[left] at (-3,5.2) {$\mathscr{A}_2$};

\draw[postaction={on each segment={mid arrow}}] (-3,5.2) -- (3,5.2);

\draw[fill=black] (3,5.2) circle (0.2cm);
\node[right] at (3,5.2) {$\mathscr{A}_3$};

\draw[postaction={on each segment={mid arrow}}] (3,5.2) -- (0,0);

\end{tikzpicture}
\end{center}
\caption{}
\end{subfigure}
\caption{An example of a task with three authorized regions $\mathscr{A}_1,\mathscr{A}_2$ and $\mathscr{A}_3$. (a) The arrangement of the regions in spacetime, notice that each region consists of two disconnected ball-shaped regions. (b) The corresponding graph of causal connections, used in the proof of theorem \ref{thm:localize} to construct the error-correcting code needed to complete the task.}
\label{fig:disconnectedtriangle}
\end{center}
\end{figure}

\begin{figure}
\begin{center}
\begin{subfigure}{.45\textwidth}
\begin{center}
\begin{tikzpicture}[scale=0.5]

	\coordinate (A) at (0,0);
	\coordinate (B) at (7,0);
	\coordinate (C) at (7,-7);
	\coordinate (D) at (0,-7);
	
	\draw (A) circle (0.2);
	\draw (B) circle (0.2);
	\draw (C) circle (0.2);
	\draw (D) circle (0.2);
	
	\draw[thick,postaction={on each segment={mid arrow}}] (0.2,0) -> (6.8,0);
	\draw[thick,postaction={on each segment={mid arrow}}] (7,-0.2) -> (7,-6.8);
	\draw[thick,postaction={on each segment={mid arrow}}] (6.8,-7) -> (0.2,-7);
	\draw[thick,postaction={on each segment={mid arrow}}] (0,-6.8) -> (0,-0.2);
	
	\draw[thick,postaction={on each segment={mid arrow}}] (6.86,-6.86) -> (1,-1);
	\draw[thick] (1,-1) -- (0.14,-0.14);
	\draw[thick,postaction={on each segment={mid arrow}}] (6.86,-0.14) -> (1,-6);
	\draw[thick] (1,-6) -- (0.14,-6.86);
	
	\node[above left] at (A) {$\mathscr{A}_1$};
	\node[above right] at (B) {$\mathscr{A}_2$};
	\node[below right] at (C) {$\mathscr{A}_3$};
	\node[below left] at (D) {$\mathscr{A}_4$};

\end{tikzpicture}
\end{center}
\caption{}
\end{subfigure}
\hfill
\begin{subfigure}{.45\textwidth}
\begin{center}
\begin{tikzpicture}[scale=0.5]
  
	\coordinate (P) at (0,-10);
	\coordinate (Q) at (7,-10);
	\coordinate (R) at (7,-17);
	\coordinate (W) at (0,-17);
	
  
	\draw (P) circle (0.2);
	\draw (Q) circle (0.2);
	\draw (R) circle (0.2);
	\draw (W) circle (0.2);
	
  
	\draw[thick] (0.2,-10) -- (6.8,-10);
	
	\draw[thick] (7,-10.2) -- (7,-16.8);
	\draw[thick] (6.8,-17) -- (0.2,-17);
	\draw[thick] (0,-16.8) -- (0,-10.2);
	\draw[thick] (6.86,-16.86) -- (0.14,-10.14);
	\draw[thick] (6.86,-10.14) -- (0.14,-16.86);
  
	\node[above left] at (P) {$\mathscr{A}_1$};
	\node[above right] at (Q) {$\mathscr{A}_2$};
	\node[below right] at (R) {$\mathscr{A}_3$};
	\node[below left] at (W) {$\mathscr{A}_4$};
	
    \draw[thick,purple] ($(-0.5,-9.5)+(1.7,0)$) arc (0:-90:1.7);
    \draw[thick,purple] ($(-0.5,-17.5)+(1.7,0)$) arc (0:90:1.7);
    \draw[thick,purple] ($(7.5,-9.5)-(1.7,0)$) arc (0:90:-1.7);
    \draw[thick,purple] ($(7.5,-17.5)-(1.7,0)$) arc (0:-90:-1.7);

\end{tikzpicture}
\end{center}
\caption{}
\end{subfigure}
\caption{Illustration of the functioning of the error-correcting code used in theorem \ref{thm:localize}. a) A directed graph that describes the causal connections between the authorized regions of a localize task. In this case the task involves four authorized regions. b) To complete the task, we employ an error-correcting code that associates a share to each edge in the corresponding undirected graph. The encoded qubit can be reconstructed from the shares associated with the edges attached to any one vertex, corresponding to the sets of edges crossed by the purple arcs. For a single logical qubit, the shares on each edge consist of two qubits. A detailed construction of the code can be found in \cite{hayden2016summoning}, and a more efficient version in \cite{wu2017efficient}. For infinite dimensional versions see \cite{hayden2016spacetime}.}
\label{fig:graphcode}
\end{center}
\end{figure}
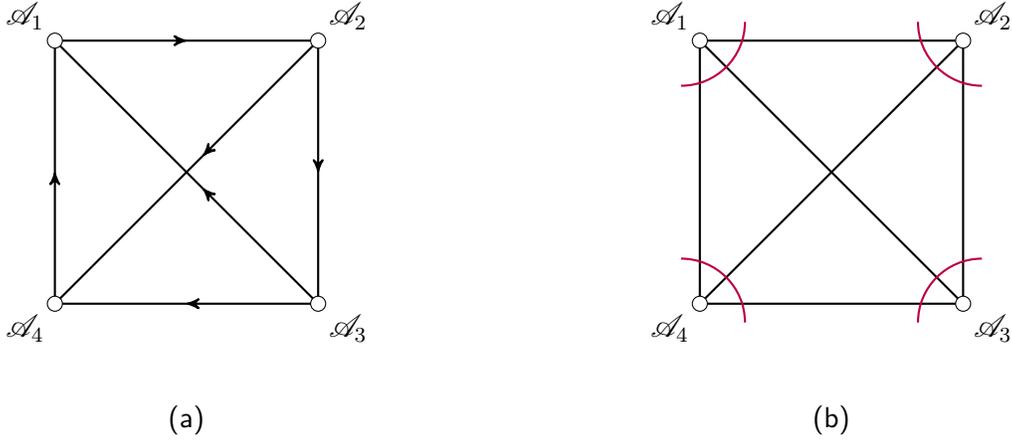

This result is particularly simple and expected from earlier work on summoning. Indeed, the conditions for summoning to a collection of diamonds are the same as for localizing to a collection of authorized regions (see \cite{hayden2016summoning}, or appendix \ref{appendix:equivalence}). 

\subsection{Localizing and excluding quantum information}

Now that we have an understanding of when and how a quantum system can be localized to many spacetime regions, we can approach the localize-exclude task. This task includes a notion of unauthorized region, a region in spacetime from which the system must be excluded in the sense described in the last section. Further, we will require that accessing an unauthorized region reveals no information about the quantum system. We collect these ideas into the following definition.
\begin{definition} A \textbf{localize-exclude task} involves two agencies, Alice and Bob, and is specified by a tuple $\{A,s,\{\mathscr{A}_1,...,\mathscr{A}_n\},\{\mathscr{U}_1,...,\mathscr{U}_m\}\}$, consisting of:
\begin{enumerate}
    \item A quantum system $A$. In general $A$ may be a subsystem of some overall pure state $\ket{\Psi}_{AR}$. The state on $AR$ is unknown to both Alice and Bob. 
    \item A start point $s$, at which Alice initially holds system $A$
    \item A collection of spacetime regions $\{\mathscr{A}_1,...,\mathscr{A}_n\}$, which we call the authorized regions
    \item A collection of spacetime regions $\{\mathscr{U}_1,...,\mathscr{U}_m\}$, which we call the unauthorized regions
\end{enumerate}
Bob will choose to access one of the $\mathscr{A}_i$ or $\mathscr{U}_i$, and will attempt to construct the quantum system $A$ from his access. Alice successfully completes the task if both (a) Bob is able to construct $A$ when he accesses any one of the $\mathscr{A}_i$ and (b) Bob learns no information about $A$ if he accesses any one of the $\mathscr{U}_i$.
\end{definition}
If Alice successfully completes the localize-exclude task, we say she has localized system $A$ to the corresponding authorized regions while excluding it from the unauthorized regions. 

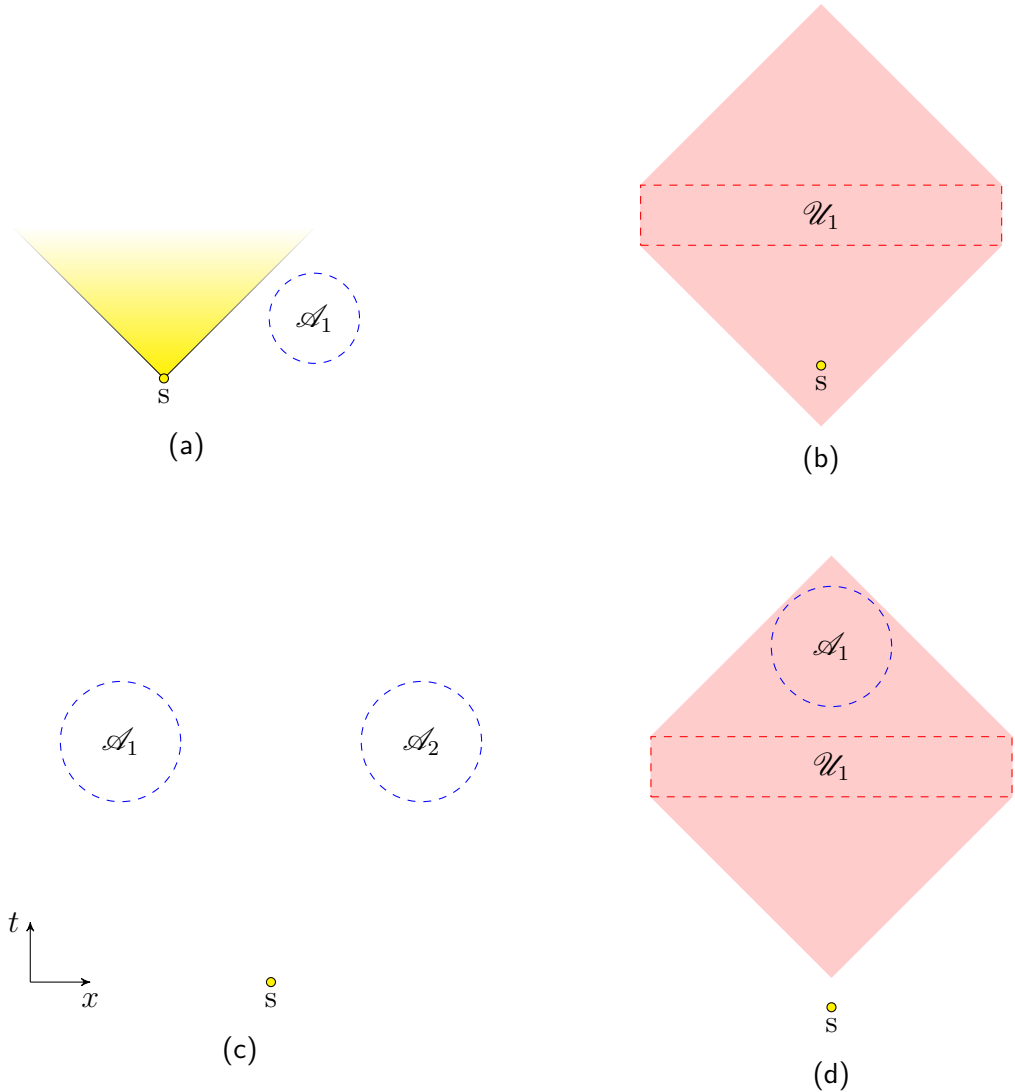
\begin{figure}
\centering
\begin{subfigure}{.45\textwidth}
\centering
\begin{tikzpicture}[scale=0.4]

\node at (0,13) {};
\filldraw[color = black, fill = yellow, path fading = north] (-5,5) -- (0,0) -- (5,5) -- (-5,5);

\node[below] at (0,0) {s};
\draw[fill=yellow] (0,0) circle (0.15cm);

\draw[dashed, blue] (5,2) circle (1.5cm);
\node at (5,2) {$\mathscr{A}_1$};
    
\end{tikzpicture}
\caption{}
\end{subfigure}
\hfill
\begin{subfigure}{.45\textwidth}
\centering
\begin{tikzpicture}[scale=0.4]

\begin{scope}[shift={($(0,-2)$)}]    
\draw[red!20!white, fill=red!20!white] (-6,8) -- (0,14) -- (6,8) -- (6,6) -- (0,0) -- (-6,6) -- (-6,8);
\end{scope}

\node at (0,13) {};

\node[below] at (0,0) {s};
\draw[fill=yellow] (0,0) circle (0.15cm);

\begin{scope}[shift={($(0,1)$)}]  
\draw[dashed,red] (-6,3) -- (6,3) -- (6,5) -- (-6,5) -- (-6,3);
\node at (0,4) {$\mathscr{U}_1$};
\end{scope}
    
\end{tikzpicture}
\caption{}
\end{subfigure}
\vspace{1cm}

\hfill
\begin{subfigure}{.45\textwidth}
  \centering
  \begin{tikzpicture}[scale=0.4]

\node at (0,13) {};

\draw[dashed, blue] (-5,8) circle (2cm);
\node at (-5,8) {$\mathscr{A}_1$};
\draw[dashed,blue] (5,8) circle (2cm);
\node at (5,8) {$\mathscr{A}_2$};

\node[below] at (0,0) {s};
\draw[fill=yellow] (0,0) circle (0.15cm);

\begin{scope}[shift={($(-5,0)$)}]       
    \draw[->] (-3,0) -- (-1,0);
    \node [below] at (-1,0) {$x$};
    \draw[->] (-3,0) -- (-3,2);
    \node [left] at (-3,2) {$t$};
\end{scope}

\end{tikzpicture}
\caption{}
\end{subfigure}
\hfill
\begin{subfigure}{.45\textwidth}
\centering
\begin{tikzpicture}[scale=0.4]

\draw[red!20!white, fill=red!20!white] (-6,8) -- (0,14) -- (6,8) -- (6,6) -- (0,0) -- (-6,6) -- (-6,8);

\node at (0,13) {};

\draw[dashed, blue] (0,11) circle (2cm);
\node at (0,11) {$\mathscr{A}_1$};

\draw[dashed,red] (-6,6) -- (6,6) -- (6,8) -- (-6,8) -- (-6,6);
\node at (0,7) {$\mathscr{U}_1$};

\node[below] at (0,-1) {s};
\draw[fill=yellow] (0,-1) circle (0.15cm);
    
\end{tikzpicture}
\caption{}
\end{subfigure}
\caption{Four impossible localize-exclude tasks: (a) An authorized region is entirely outside the future light cone of $s$, so system $A$ can't be localized there without violating the no-signalling principle. (b) The initial location of the quantum system is in the domain of dependence of an unauthorized region $\mathscr{U}_1$, so can be reconstructed from data in $\mathscr{U}_1$. (c) A quantum system cannot be localized to both the spacetime regions $\mathscr{A}_1$ and $\mathscr{A}_2$, due to the no-cloning theorem. (d) A quantum system cannot be localized to $\mathscr{A}_1$ without passing through the region $\mathscr{U}_1$, since there is no causal curve which passes through $\mathscr{A}_1$ and not $\mathscr{U}_1$. The red shaded region indicates the domain of dependence of the unauthorized region $\mathscr{U}_1$. The yellow shading indicates the future light cone of the start point.}
\label{fig:obvious}
\end{figure}

As an initial approach to understanding the localize-exclude task we can list off the most basic restrictions that we expect to apply. First, the two restrictions occurring in the context of the localize task are still relevant: the start point should have a point from each authorized region in its future light cone, and there should be no causally disjoint pairs of authorized regions. There are also additional restrictions relating to the unauthorized regions however. In particular, we can never have an authorized region $\mathscr{A}_i$ be contained in the domain of dependence of an unauthorized region $\mathscr{U}_j$, since then all information which enters $\mathscr{A}_i$ also enters $\mathscr{U}_j$. Finally, the start point too should not be contained in the domain of dependence of any unauthorized region. We illustrate each these conditions in figure \ref{fig:obvious}. Remarkably, a localize-exclude task $\{A,s,\{\mathscr{A}_1,...,\mathscr{A}_n\},\{\mathscr{U}_1,...,\mathscr{U}_m\}\}$ will turn out to be possible to complete so long as none of the four situations in figure \ref{fig:obvious} occur. 

\begin{figure}
  \centering
  \begin{tikzpicture}[scale=0.45]

\draw[dashed,blue] (5,15) circle (2cm);
\node at (5,15) {$\mathscr{A}_2$};

\draw[dashed,blue] (0,3) circle (2cm);
\node at (0,3) {$\mathscr{A}_1$};

\node[below right] at (0.5,-6) {$A$};
\draw[fill=yellow] (0.5,-6) circle (0.15cm);

\begin{scope}[shift={($(2,-2)$)}]   
\draw[dashed,red] (4,5) -- (6,7) -- (-4,17) -- (-6,15) -- (4,5);
\node[above left] at (0,11) {$\mathscr{U}$};
\end{scope}

\begin{scope}[shift={($(-5,-5)$)}]       
    \draw[->] (-4,0) -- (-2,0);
    \node [below] at (-2,0) {$x$};
    \draw[->] (-4,0) -- (-4,2);
    \node [left] at (-4,2) {$t$};
\end{scope}

\begin{scope}[shift={($(0,-5)$)}]  
\draw[black] (-1,1) -- (1,1) -- (1,3) -- (-1,3) -- (-1,1);
\node at (0,2) {$\mathcal{U}_k$};
\draw[postaction={on each segment={mid arrow}}] (0.5,-1) -- (0.5,1);
\end{scope}

\draw[double,postaction={on each segment={mid arrow}}] (9.5,-7) to [out=90,in=-70] (6.5,13.5);
\node[below] at (9.5,-7) {$k$}; 

\draw[double,postaction={on each segment={mid arrow}}] (-1,-7) to [out=90,in=-90] (-0.5,-4);
\node[below] at (-1,-7) {$k$};

\draw[double,postaction={on each segment={mid arrow}}] (-0.5,-2) to [out=90,in=-50] (-6,13);

\draw[postaction={on each segment={mid arrow}}] (0.5,-2) to  [out=90,in=-110] (4,13.4);

\end{tikzpicture}
\caption{Illustration of the protocol for completing a localize-exclude task with two authorized regions and one unauthorized region that satisfy the conditions of theorem \ref{thm:localexclude}. In the distant past, Alice prepares copies of the classical string $k$. She brings one copy of $k$ to each of $\mathscr{A}_1$ and $\mathscr{A}_2$ along a path which does not cross $\mathscr{U}$ --- this is always possible by condition (iii). She must also bring the classical string to the start point $s$, and encode the $A$ system using the quantum one-time pad \cite{ambainis2000private}. The overall state on $A$ and its purifying system $R$ is then of the form $(\mathcal{U}_k\otimes \mathcal{I} )\ket{\Psi}_{AR}$. The encoded system $A$ is sent through both authorized regions. By following this protocol both authorized regions contain $k$ and the encoded $A$ system, while the unauthorized region contains the encoded system only.}
\label{fig:onetimepad}
\end{figure}
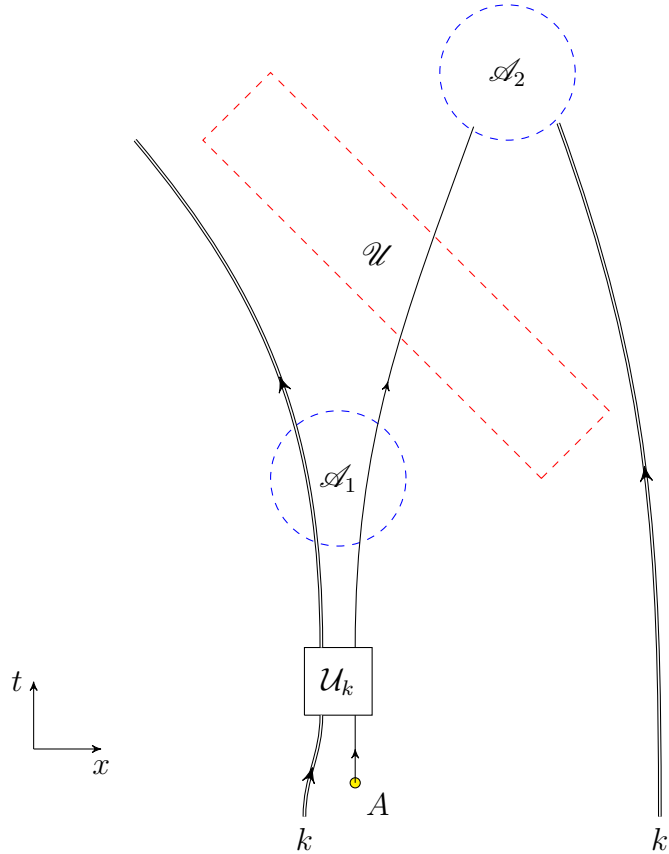

As a warm-up to the general case, consider the example given in the introduction as figure \ref{fig:firstexample}. There, a single unauthorized region blocks the path between two authorized ones. As we illustrate in figure \ref{fig:onetimepad}, it is nonetheless possible to complete the task using the quantum one-time pad \cite{ambainis2000private}. Near the start point, a unitary $\mathcal{U}_k$ is applied to $A$ with $k$ chosen at random. The overall pure state is then $\mathcal{U}_k\otimes \mathcal{I}\ket{\Psi}_{AR}$. To an observer who is unaware of the key $k$, the density matrix of the state is $\rho_{AR}=\sum_k\frac{1}{|k|}(\mathcal{U}_k\otimes \mathcal{I})\ketbra{\Psi}{\Psi}(\mathcal{U}^\dagger_k\otimes \mathcal{I})$. By carefully choosing the set of possible unitaries $\mathcal{U}_k$, one can arrange that $\rho_{AR} = \mathcal{I}_A/d_A\otimes \rho_R$, so that Bob has learned nothing about the $A$ system whenever he does not learn $k$. This is possible when $A$ consists of $n$ qubits and $k$ consists of $4n$ bits \cite{ambainis2000private}. Once encoded using the one-time pad, the $A$ system is sent through both authorized regions by allowing it to pass through the unauthorized region. An access to the unauthorized region then only sees the maximally mixed state. The classical key $k$ is also sent to both authorized regions, but along trajectories that avoid the unauthorized one.

A similar technique can be applied to the general case of many authorized and many unauthorized regions. As we show in the proof of theorem \ref{thm:localexclude} given below, the strategy is to first encode the $A$ system into an error-correcting code so that it can be localized to each authorized region. Then each share in that error-correcting code is encoded using a classical string and the quantum one-time pad. We then leverage classical secret sharing to allow us to get the encoding string to the needed authorized regions while avoiding all the unauthorized regions.

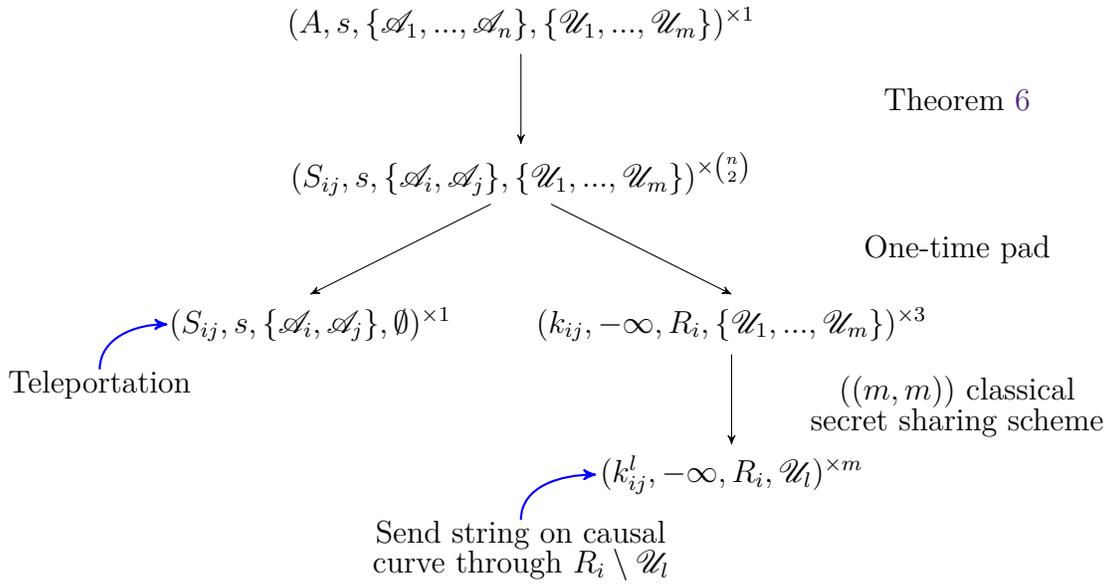
\begin{figure}
\begin{center}
\begin{tikzpicture}[scale=0.4]

\node at (0,15) {$(A,s,\{\mathscr{A}_1,...,\mathscr{A}_n\},\{\mathscr{U}_1,...,\mathscr{U}_m\})^{\times 1}$};

\draw[->] (0,14) -- (0,11);

\node at (0,10) {$(S_{ij},s,\{\mathscr{A}_i,\mathscr{A}_j\},\{\mathscr{U}_1,...,\mathscr{U}_m\})^{\times \binom{n}{2}}$};
\node at (-7,5) {$(S_{ij},s,\{\mathscr{A}_i,\mathscr{A}_j\},\emptyset)^{\times 1}$};

\node at (-14,3) {\textcolor{black}{Teleportation}};
\draw[thick,->,blue] (-14,3.5) to [out=90,in=180] (-11.75,5);

\node at (7,5) {$(k_{ij},-\infty,R_i,\{\mathscr{U}_1,...,\mathscr{U}_m\})^{\times 3}$};
\node at (7,0) {$(k_{ij}^l,-\infty,R_i,\mathscr{U}_l)^{\times m}$};

\node at (0,-2) {\textcolor{black}{Send string on causal}};
\node at (0,-3) {\textcolor{black}{curve through $R_i\setminus \mathscr{U}_l$}};
\draw[thick,->,blue] (0,-1.5) to [out=90,in=180] (2.5,0);

\draw[->] (-1,9) -- (-7,6);
\draw[->] (1,9) -- (7,6);
\draw[->] (7,4) -- (7,1);

\node[black] at (14.5,12.5) {Theorem \ref{thm:localize}};
\node[black] at (14.5,7.5) {One-time pad};
\node[black] at (14.5,2.75) {$((m,m))$ classical};
\node[black] at (14.5,1.75) {secret sharing scheme};

\end{tikzpicture}
\end{center}
\caption{Diagram of the sufficiency proof of theorem \ref{thm:localexclude}. In three steps, the proof reduces completing the localization task on the system $A$ with $n$ authorized sets and $m$ unauthorized sets, denoted by $(A,s,\{\mathscr{A}_1,...,\mathscr{A}_n\},\{\mathscr{U}_1,...,\mathscr{U}_m\})$, to completing $\binom{n}{2}$ instances of $(S_{ij},s,\{\mathscr{A}_i,\mathscr{A}_j\},\emptyset)$ on quantum shares, and $3 m \binom{n}{2}$ instances of $(k_{ij}^l,-\infty,R_i,\mathscr{U}_l)$ on classical shares, where the region $R_i$ may be either the start point or an authorized region. The notation $-\infty$ indicates the share is available at early times. The first step in the protocol is to recycle the error-correcting code from theorem \ref{thm:localize} to encode the $A$ system into shares $S_{ij}$. At the second step, the one-time pad is applied to each of the $S_{ij}$. This allows the unauthorized regions to be avoided by introducing additional classical shares, but without the need for further uses of quantum error-correcting codes.}
\label{fig:proofdiagram}
\end{figure}

\begin{figure}
\begin{center}
  \begin{tikzpicture}[scale=0.4]

\draw[red!20!white, fill=red!20!white] (-15,12) -- (-15,14) -- (-10,19) -- (-5,14) -- (-5,12) -- (-10,7) -- (-15,12);
\draw[dashed,red] (-15,12) -- (-5,12) -- (-5,14) -- (-15,14) -- (-15,12);
\node at (-10,13) {$\mathscr{U}_1$};

\draw[red!20!white, fill=red!20!white] (15,12) -- (15,14) -- (10,19) -- (5,14) -- (5,12) -- (10,7) -- (15,12);
\draw[dashed,red] (15,12) -- (5,12) -- (5,14) -- (15,14) -- (15,12);
\node at (10,13) {$\mathscr{U}_2$};

\draw[dashed,blue] (-10,10) circle (1.5cm);
\node at (-10,10) {$\mathscr{A}_1$};
\draw[dashed,blue] (-10,16) circle (1.5cm);
\node at (-10,16) {$\mathscr{A}_2$};

\draw[dashed,blue] (10,10) circle (1.5cm);
\node at (10,10) {$\mathscr{A}_2$};
\draw[dashed,blue] (10,16) circle (1.5cm);
\node at (10,16) {$\mathscr{A}_1$};

\draw[blue,postaction={on each segment={mid arrow}}] (0,-8) -- (0,-4);
\node[below] at (0,-8) {$A$};
\draw[fill=yellow] (0,-8) circle (0.15cm);

\begin{scope}[shift={($(-11,-15)$)}]       
    \draw[->] (-3,0) -- (-2,0);
    \node [right] at (-2,0) {$x$};
    \draw[->] (-3,0) -- (-3,1);
    \node [left] at (-3,1) {$t$};
\end{scope}

\begin{scope}[shift={($(0,-5)$)}]  
\draw[black] (-1,1) -- (1,1) -- (1,3) -- (-1,3) -- (-1,1);
\node at (0,2) {$\mathcal{U}_k$};
\end{scope}

\draw[double,postaction={on each segment={mid arrow}}] (11,-14) to [out=90,in=-90] (11,20);
\node[below] at (10.5,-14) {$k_2$}; 

\draw[double,postaction={on each segment={mid arrow}}] (10,-14) to [out=135,in=-80] (0.5,-4);

\draw[double,postaction={on each segment={mid arrow}}] (-11,-14) to [out=90,in=-90] (-11,20);
\node[below] at (-10.5,-14) {$k_1$}; 

\draw[double,postaction={on each segment={mid arrow}}] (-10,-14) to [out=45,in=-100] (-0.5,-4);

\draw[blue,postaction={on each segment={mid arrow}}] (-0.5,-2) to [out=100,in=-90] (-9,10);
\draw[blue,postaction={on each segment={mid arrow}}] (-9,10) to [out=90,in=-90] (-9,20);

\end{tikzpicture}
\end{center}
\caption{An example of a localize-exclude task and illustration of the protocol provided by theorem \ref{thm:localexclude} for its completion. Near the start point the system $A$ is encoded using the quantum one-time pad and sent (along the blue curve) through both authorized regions. The string $k$ satisfies $k=k_1\oplus k_2$, so that $k_1,k_2$ form the two shares of a $((2,2))$ secret sharing scheme. $k_1$ is sent through $\mathscr{A}_1$ and $\mathscr{A}_2$ while avoiding $\mathscr{U}_1$, while $k_2$ is sent through $\mathscr{A}_1$ and $\mathscr{A}_2$ while avoiding $\mathscr{U}_2$. Consequently, each $\mathscr{A}_i$ contains all of the classical shares $k_i$ along with the encoded $A$ system, while each $\mathscr{U}_i$ is missing one $k_i$.}
\label{fig:LEexample}
\end{figure}
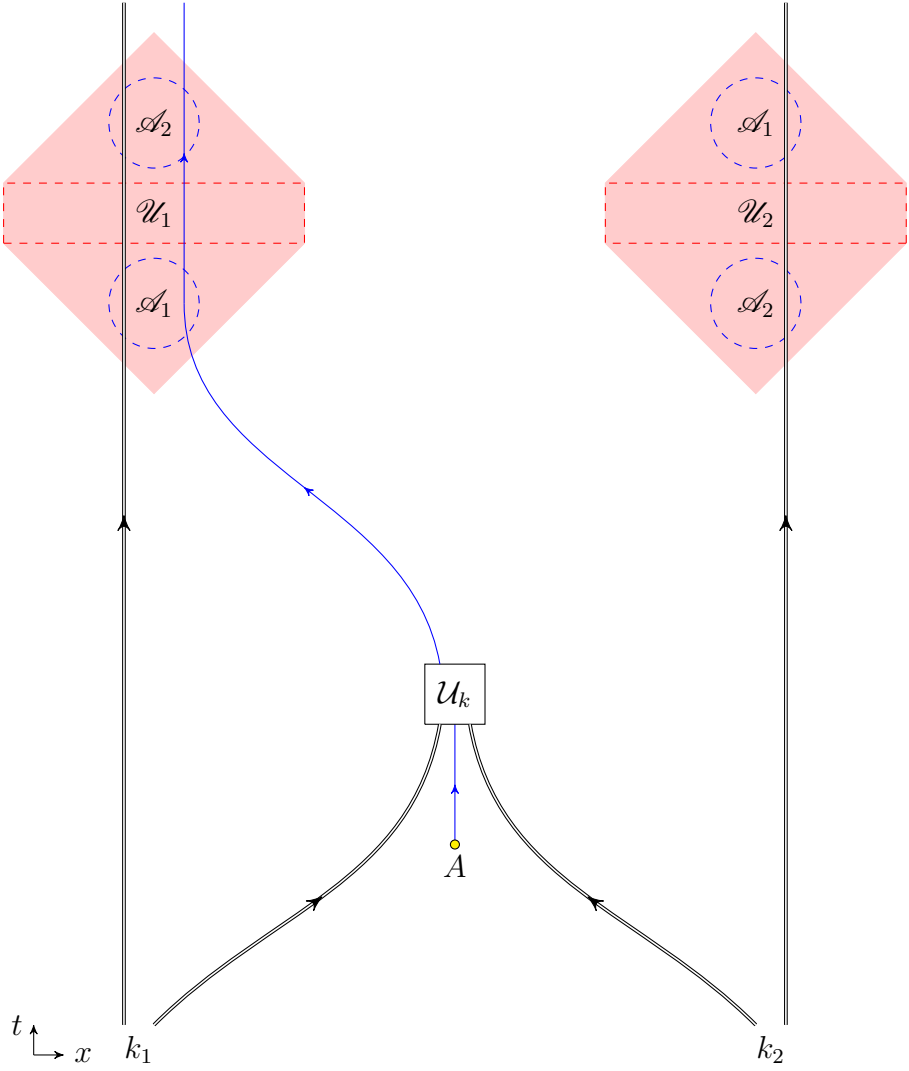

We are now ready to state theorem \ref{thm:localexclude} and give the proof. The proof of sufficiency is somewhat lengthy, so we have provided figure \ref{fig:proofdiagram} which summarizes the key steps taken.
\begin{theorem}\label{thm:localexclude}
Given a collection of authorized regions $\{\mathscr{A}_1,...,\mathscr{A}_n\}$, unauthorized regions $\{\mathscr{U}_1,...,\mathscr{U}_m\}$, and start point $s$, a localize-exclude task is possible if and only if the following three conditions are satisfied.
\begin{enumerate*}
\item The starting location of the system $A$ (a) has at least one point from each authorized region in its causal future, and (b) is not in the domain of dependence of any unauthorized region.
\item Every pair of authorized regions $(\mathscr{A}_i,\mathscr{A}_j)$ are causally connected.
\item For every pair $(\mathscr{A}_i,\mathscr{U}_j)$ of authorized and unauthorized regions, $\mathscr{A}_i$ is not contained in the domain of dependence of $\mathscr{U}_j$.
\end{enumerate*}
\end{theorem}
\begin{proof}\,The necessity of conditions $(i)(a)$ and $(ii)$ follow from the same arguments as in theorem \ref{thm:localize}. To argue the necessity of condition $(iii)$, notice that if $\mathscr{A}_i$ is contained in the domain of dependence of $\mathscr{U}_j$, then the state of the quantum fields within $\mathscr{A}_j$ is determined by unitary evolution from the fields within $\mathscr{U}_i$. Then whenever the $A$ system can be determined from $\mathscr{A}_i$ it is also possible to recover it from $\mathscr{U}_j$. Condition $(i)(b)$ is necessary for the same reason.

To demonstrate sufficiency we construct an explicit protocol to complete the task in the case where all three conditions are true. It is useful to recall the notation $(A,s,\{\mathscr{A}_1,...,\mathscr{A}_n\},\{\mathscr{U}_1,...,\mathscr{U}_m\})$, which describes a localize-exclude task by specifying the system on which we must complete the task, the start point, authorized regions, and unauthorized regions. As a first step in constructing our protocol, we encode the system $A$ into the error-correcting code used in theorem \ref{thm:localize}. Using this code and localizing each share in the code to its two associated authorized regions would localize the system to each authorized region. However, here we also need to exclude the system from all of the unauthorized regions. To do this, we will localize each share $S_{ij}$ to $\mathscr{A}_i$ and $\mathscr{A}_j$ while also avoiding every unauthorized region. In other words, encoding $A$ into the codeword stabilized code reduces completing the original task to completing the tasks $(S_{ij},s,\{\mathscr{A}_i,\mathscr{A}_j\},\{\mathscr{U}_1,...,\mathscr{U}_m\})$ for every share $S_{ij}$.

By using the quantum one-time pad and classical secret sharing it is possible to further reduce completing the $(S_{ij},s,\{\mathscr{A}_i,\mathscr{A}_j\},\{\mathscr{U}_1,...,\mathscr{U}_m\})$ task. In particular, at $s$ use the quantum one-time pad to encode the share $S_{ij}$ using some classical string $k_{ij}$. We may freely send the encoded share through $\mathscr{A}_i$ and $\mathscr{A}_j$ so long as the classical string $k_{ij}$ is kept out of all of the unauthorized regions, and is made available at $s, \mathscr{A}_i$, and $\mathscr{A}_j$. Thus, the task $(S_{ij},s,\{\mathscr{A}_i,\mathscr{A}_j\},\{\mathscr{U}_1,...,\mathscr{U}_m\})$ is equivalent to completing $(S_{ij},s,\{\mathscr{A}_i,\mathscr{A}_j\},\emptyset)$ along with $(k_{ij},-\infty, \{s,\mathscr{A}_i,\mathscr{A}_j\},\{\mathscr{U}_1,...,\mathscr{U}_m\} )$\footnote{We've introduced the notation $-\infty$ to indicate the start point is located in the distant past. This is the appropriate task to consider completing on the classical system $k_{ij}$ as Alice may prepare these strings at some early time.}. 

To finish the protocol, we first notice that theorem \ref{thm:localize} shows that we can complete any task of the form $(S_{ij},s,\{\mathscr{A}_i,\mathscr{A}_j\},\emptyset)$ given that conditions $(i)(a)$ and $(ii)$ hold. The task $(k_{ij},-\infty, \{s,\mathscr{A}_i,\mathscr{A}_j\},\{\mathscr{U}_1,...,\mathscr{U}_m\} )$ is also easily handled. Note that since the task is to be completed on a classical string, we can produce three copies of $k_{ij}$ and worry separately about sending the string to $s$ and each of $\mathscr{A}_i$ and $\mathscr{A}_j$, so we have to complete three instances of $(k_{ij},-\infty, R,\{\mathscr{U}_1,...,\mathscr{U}_m\} )$, where $R$ can be $s, \mathscr{A}_i$ or $\mathscr{A}_j$. To complete these, encode $k_{ij}$ into an $((m,m))$ secret sharing scheme\footnote{A $((k,n))$ secret sharing scheme is one where any $k$ of the $n$ total shares can be used to reconstruct the secret while any $k-1$ shares reveal nothing about the secret. A $((m,m))$ scheme is the appropriate one here because we want every share to be needed to reconstruct $k_{ij}$.} with shares $k_{ij}^l$. Then complete the tasks $(k_{ij}^l,-\infty, R, \mathscr{U}_l )$. This completes the task with all $m$ unauthorized regions since the classical string is kept out of $\mathscr{U}_l$ so long as at least one of the shares in the $((m,m))$ scheme is. 

It remains to complete the tasks of the form $(k_{ij}^l,-\infty, R,\mathscr{U}_l)$. When $R$ is one of the authorized sets, condition (iii) guarantees that $R$ is not in the domain of dependence of $\mathscr{U}_l$, which means there is a causal curve passing through $R$ which does not enter $\mathscr{U}_l$. To complete the task, simply send $k_{ij}^l$ along this curve. When $R$ is the start point $s$, condition $(i)(b)$ guarantees there is a causal curve passing through $s$ and not $\mathscr{U}_l$, so again we can complete this task. 
\end{proof} 

An example of the protocol used in this proof is given as figure \ref{fig:LEexample}.

Earlier we mentioned the similarity of conditions $(ii)$ and $(iii)$ to corresponding conditions for quantum secret sharing. A quantum secret sharing scheme \cite{gottesman2000theory} is specified by an access structure, with the access structure consisting of subsets of parties deemed authorized and subsets deemed unauthorized. A quantum secret sharing scheme can be constructed under two conditions \cite{gottesman2000theory}: (a) (no-cloning) no two authorized sets can be disjoint and (b) (monotonicity) no authorized set can be contained within an unauthorized set. Conditions $(ii)$ and $(iii)$ of the localize-exclude theorem are exactly these conditions rephrased in a context appropriate to spacetime. 

\begin{figure}
    \centering
    \begin{tikzpicture}
    
    \draw (-5,3) circle (1);
    \node at (-5,3) {$\Sigma_1$};
    
    \draw (0,3) circle (1);
    \node at (0,3) {$\Sigma_2$};
    
    \draw (5,3) circle (1);
    \node at (5,3) {$\Sigma_3$};
    
    \node[below] at (0,-2) {s};
    \draw[fill=yellow] (0,-2) circle (0.1);
    
    \end{tikzpicture}
    \caption{Example of the embedding of a secret sharing scheme with arbitrary access structure into a localize-exclude task. We consider a secret sharing scheme that involves three parties, and has authorized sets $S_1=\{1,2\}$, $S_2=\{2,3\}$ and $S_3=\{1,2,3\}$, with all other subsets of parties deemed unauthorized. In the corresponding localize-exclude task, the three parties become three causally disjoint spacetime regions $\Sigma_1, \Sigma_2$ and $\Sigma_3$. Further, this localize-exclude task has authorized regions $\mathscr{A}_1 = \Sigma_1\cup\Sigma_2$, $\mathscr{A}_2=\Sigma_2\cup\Sigma_3$ and $\mathscr{A}_3=\Sigma_1\cup \Sigma_2\cup \Sigma_3$. The start point $s$ has been placed at an early enough time that all the $\Sigma_i$ are in its future light cone.}
    \label{fig:secretsharingreduction}
\end{figure}
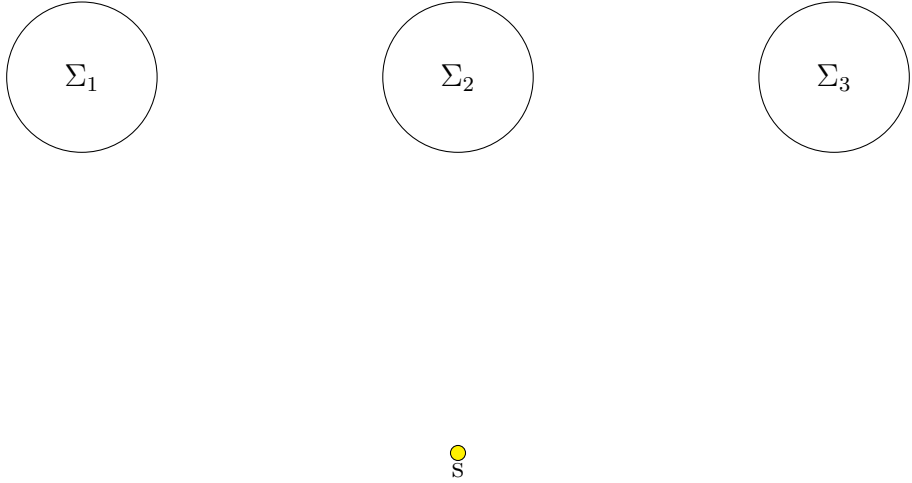

Beyond this similarity, we can embed any secret sharing scheme into a localize-exclude task. Consider $n$ parties, Bob$_1$,...,Bob$_n$, who each can potentially access an associated spacetime region $\Sigma_i$. Take the authorized and unauthorized regions to consist of unions of $\Sigma_i$'s so that a full authorized region $\mathscr{A}_i$ can be accessed only if some collection of Bobs agree to cooperate. Choose the regions $\Sigma_i$ to be all causally disjoint. In this setting two authorized regions being causally connected occurs if and only if they share a $\Sigma_i$. Then condition $(ii)$ of theorem \ref{thm:localexclude}, which requires causal connections between authorized regions, reduces to the requirement that every pair of authorized regions share at least one $\Sigma_i$. This is exactly the no-cloning requirement on secret sharing. Further, condition $(iii)$ reduces to no $\mathscr{U}_i = \Sigma_{i1}\cup...\cup\Sigma_{in}$ containing as a subset some $\mathscr{A}_j=\Sigma_{j1}\cup...\cup\Sigma_{j2}$ under the same restriction of having causally disjoint $\Sigma_i$. This is just the monotonicity condition on quantum secret sharing schemes. Finally, to embed our quantum secret sharing task into a localize-exclude task we should ensure that condition $(i)$ becomes trivial, which we can do by sending the start point $s$ to an early time. We illustrate the embedding of a secret sharing task into a localize-exclude task in figure \ref{fig:secretsharingreduction}.

Theorem \ref{thm:localexclude} shows that completing a localize-exclude task with unauthorized regions requires only the same quantum error-correcting code as used in the case with no unauthorized regions. Hiding the system from the unauthorized regions can be accomplished using only the quantum one-time pad and classical secret sharing. This is similar to the approach taken in \cite{javelle2012new}, where quantum error-correcting codes are combined with the quantum one-time pad to yield quantum secret sharing schemes. By using the efficient error-correcting code underlying our protocol however, we arrive at a particularly efficient construction of quantum secret sharing schemes. In particular we find that there is a universal quantum error-correcting code with $2\binom{n}{2}$ shares for $n$ the number of authorized sets which, along with uses of the one-time pad and classical secret sharing, constructs quantum secret sharing schemes with arbitrary access structures. Using Shamir's method~\cite{shamir1979share} to construct the classical secret sharing schemes, the $3m\binom{n}{2}$ instances of the $((m,m))$ classical scheme will each require $O(m \log m)$ bits, where $m$ was the number of unauthorized sets. In total, $O(n^2)$ qubits and $O(m^2n^2 \log m)$ classical bits are used in the localize-exclude construction. This provides the first construction of quantum secret sharing schemes using a number of qubits polynomial in the number of authorized sets. Previously, efficient constructions were known for threshold schemes and certain other special access structures. (See, \emph{e.g.}~\cite{markham2008graph,sarvepalli2010matroids,javelle2012new}.) Since the number of unauthorized sets can grow exponentially with $n$, the classical bits used can be exponentially large. This is to be expected since it is conjectured to be impossible to construct classical secret sharing schemes for arbitrary access structures without consuming exponential resources~\cite{beimel2011secret}. 

\section{State-assembly}\label{sec:assembly}

\subsection{State-assembly with authorized regions}

In the localize-exclude task Bob can access any one of a set of spacetime regions. Alice, who holds various quantum systems within those regions, is helpless to prevent Bob's access. In an alternative scenario we can have Bob request information from Alice. Alice is free to comply with the request or to reject it, and hand over no information. Certain sets of requests are deemed authorized, others unauthorized. Sets of requests corresponding to authorized sets should result in Alice handing over sufficient information for the system to be reconstructed; requests to unauthorized sets should reveal no information about the system. Considering such scenario's leads us to construct the \emph{state-assembly task}. 

Before giving a precise definition of the task we introduce a few constructions. To specify locations where Bob may request the system we designate certain spacetime points as call points $c_i$. At each call point a bit $b_i\in \{0,1\}$ is revealed to Alice. To each call point there corresponds a return point $r_i$. Together, a call point and the corresponding reveal point define a causal diamond.
\begin{definition}\label{def:diamonds} 
The \textbf{causal diamond} $D_i$ is defined as the intersection of the points in the past light cone of $r_i$ with those in the future light cone of $c_i$.
\end{definition}
If $b_i=1$ we say the diamond $D_i$ has been called to. The causal diamond represents the spacetime region in which it is possible to both know that a call was received, and to use this information to influence what is handed over at the corresponding return point.

We can now define the state-assembly task.
\begin{definition}\label{def:stateassembly}
A \textbf{state-assembly} task involves two agencies, Alice and Bob, and is specified by a tuple $\{A,s,\{\mathscr{A}_1,...,\mathscr{A}_n\},\{\mathscr{U}_1,...,\mathscr{U}_m\}\}$, consisting of
\begin{enumerate}
    \item A quantum system $A$. In general $A$ may be a subsystem of some overall pure state $\ket{\Psi}_{AR}$. The state on $AR$ is known to Alice and unknown to Bob.
    \item A start point $s$ at which Alice initially holds $A$.
    \item A collection of authorized sets of diamonds $\{\mathscr{A}_1,...,\mathscr{A}_n\}$. Each authorized set consists of a collection of diamonds, $\mathscr{A}_i = \{D_{1i},...,D_{ki}\}$.
    \item A collection of unauthorized sets of diamonds $\{\mathscr{U}_1,...,\mathscr{U}_m\}$. Each authorized set consists of a collection of diamonds, $\mathscr{U}_i = \{D_{1i},...,D_{ki}\}$.
\end{enumerate}
Alice will receive calls at a subset of the $D_i$. If the set of called to diamonds corresponds to an authorized set, Alice should return quantum systems and classical instructions sufficient to reconstruct $A$ at the associated reveal points $r_i$. If the set of calls corresponds to an unauthorized set, the systems she hands over should reveal no information about $A$. Further, no set of calls should result in Alice returning systems sufficient to construct two copies of the system.
\end{definition}
There are a few points to clarify regarding this definition. First, Alice need not hand the system over at any one of the called to diamonds. Instead the systems she hands over at the called to diamonds should together be sufficient to recover the $A$ system. Second, calls to sets of diamonds not specified as authorized or unauthorized may result in the system being handed over --- Alice still completes the task successfully so long as she does not hand over two copies of the system.

In state-assembly Alice knows the state $\ket{\Psi}_{AR}$ and can potentially prepare many copies of the $A$ system. This differs from the localize-exclude task and earlier work on summoning. However, we have also required that she never hand over more than one copy of $A$. As discussed in more detail in appendix \ref{appendix:equivalence}, this is actually leads to conditions on state-assembly that are equivalent to having Alice hold an unknown state. We have chosen to discuss this task from the perspective of a known quantum state however as it is more natural in the context of the application given in section \ref{sec:partyindependenttransfer}.

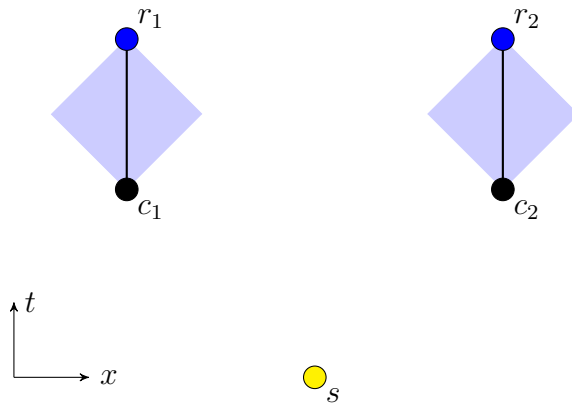
\begin{figure}
\begin{center}
\begin{tikzpicture}[scale=1]

\coordinate (P) at (0cm, 0cm);

\coordinate (C2) at (2.5cm, 2.5cm);
\coordinate (Q2) at (2.5cm, 4.5cm);
\coordinate (L2) at (3.5cm,3.5cm);
\coordinate (R2) at (1.5cm,3.5cm);

\coordinate (C1) at (-2.5cm,2.5cm);
\coordinate (Q1) at (-2.5cm,4.5cm);
\coordinate (L1) at (-3.5cm,3.5cm);
\coordinate (R1) at (-1.5cm,3.5cm);

\draw[blue!20!white, fill=blue!20!white] (C2) -- (L2) -- (Q2) -- (R2);
\draw[blue!20!white, fill=blue!20!white] (C1) -- (L1) -- (Q1) -- (R1);

\draw[thick] (C1) -- (Q1); 	 
\draw[thick] (C2) -- (Q2); 

\draw[fill=yellow] (P) circle (0.15cm);
\node[below right] at (P) {$s$};

\draw[fill=black] (C1) circle (0.15cm);
\node [below right] at (C1) {$c_1$};
\draw[fill=black] (C2) circle (0.15cm);
\node [below right] at (C2) {$c_2$};
	
\draw[fill=blue] (Q1) circle (0.15cm);
\node [above right] at (Q1) {$r_1$};
\draw[fill=blue] (Q2) circle (0.15cm);
\node [above right] at (Q2) {$r_2$};

\begin{scope}[shift={(-1,0)}]       
    \draw[->] (-3,0) -- (-2,0);
    \node [right] at (-2,0) {$x$};
    \draw[->] (-3,0) -- (-3,1);
    \node [right] at (-3,1) {$t$};
\end{scope}

\end{tikzpicture}
\end{center}
\caption{A state-assembly task with two call-return pairs. A call to $c_1$ is required to result in the system returned at $r_1$, and likewise for $c_2$ and $r_2$ (indicated by the black lines), while a call to both shouldn't result in more than one copy of the system being turned over. This task is impossible as shown by theorem \ref{thm:noassembly}, because $r_2$ is outside the future light cone of $c_1$ and $r_1$ is outside the future light cone of $c_2$. In the language of definitions \ref{def:diamonds} and \ref{def:stateassembly}, $c_1$ and $r_1$ form a causal diamond $D_1$ (shown in blue), and the authorized set $\mathscr{A}_1$ consists of the single diamond $D_1$ (similarly for $c_2$ and $r_2$).}\label{fig:noassembly}
\end{figure}

Before discussing more general constructions we begin with the simplest state-assembly task, illustrated in figure \ref{fig:noassembly}, and prove a no-assembly theorem. In this scenario there are just two authorized sets $\mathscr{A}_1$ and $\mathscr{A}_2$. 
\begin{theorem}\label{thm:noassembly}
Consider a state-assembly task with authorized sets $\mathscr{A}_1$ and $\mathscr{A}_2$ which are causally disconnected. Then this assembly task is impossible to complete with a perfect success rate.
\end{theorem}
\begin{proof}\,For Alice to successfully complete the assembly task, she must have a protocol which
\begin{enumerate*}
    \item Returns sufficient information to construct the system when $\mathscr{A}_1$ or $\mathscr{A}_2$ receive calls.
    \item Hand over information sufficient to construct at most one copy of the system for any set of calls.  
\end{enumerate*}
We can straightforwardly show that any protocol which satisfies the first requirement cannot satisfy the second, and consequently there is no such successful protocol. Indeed, suppose both $\mathscr{A}_1$ and $\mathscr{A}_2$ receive calls. Then since $\mathscr{A}_1$ and $\mathscr{A}_2$ are causally disjoint Alice's agents at the diamonds in $\mathscr{A}_1$ cannot distinguish this situation from one where only $\mathscr{A}_1$ has been called to. By (i) then they hand in sufficient information to construct the system. Similarly, Alice's agents at $\mathscr{A}_2$ will also hand in sufficient information to construct the system. Since Bob may now construct two copies of $A$, (ii) is violated. 
\end{proof}

We see that completing the assembly task to causally separated regions is impossible. Notice that it is essential that the Bobs may give calls to both diamonds: the possibility of a call to $\mathscr{A}_1 \cup \mathscr{A}_2$ along with the requirement that Alice allow assembly of not more than one copy of the system leads to Alice being unable to complete the task successfully.

Next, we look at a wider class of assembly tasks involving an arbitrary number of authorized sets $\{\mathscr{A}_i\}$.
\begin{theorem}\label{thm:authorizedassembly}
An assembly task with authorized sets $\{\mathscr{A}_1,...,\mathscr{A}_n\}$ and start point $s$ can be completed with a perfect success rate if and only if the following conditions hold.
\begin{enumerate}
\item The return point of at least one diamond from each authorized set is in the causal future of the start point.
\item Every pair of authorized sets $(\mathscr{A}_i,\mathscr{A}_j)$ are causally connected.
\end{enumerate}
\end{theorem}
\begin{proof}\,The first condition is necessary by no-signalling. The necessity of the second condition follows from the same argument as given in theorem \ref{thm:noassembly}. 

We can use theorem \ref{thm:localexclude} to show sufficiency of these conditions. There, we constructed an explicit protocol that localizes the system to each authorized region. In particular, the system is recorded as classical teleportation data and shares in a quantum error-correcting code. To complete the assembly task then, Alice should execute the localization protocol from theorem  \ref{thm:localexclude}, with the authorized sets of diamonds considered as authorized regions. Then to complete the assembly task Alice need only hand over the classical and quantum data in $\mathscr{A}_i$ when she receives calls there. 

Notice that this protocol automatically ensures Bob cannot give calls to receive two copies of the system, since Alice only uses one copy of $A$. 
\end{proof}

\subsection{State-assembly with authorized and unauthorized regions}

\begin{figure}
\centering
\begin{subfigure}{.5\textwidth}
  \centering
  \begin{tikzpicture}[scale=0.6]

\coordinate (P) at (0cm, 0cm);

\coordinate (C2) at (2.5cm, 2.5cm);
\coordinate (Q2) at (2.5cm, 4.5cm);
\coordinate (L2) at (3.5cm,3.5cm);
\coordinate (R2) at (1.5cm,3.5cm);

\coordinate (C1) at (-2.5cm,2.5cm);
\coordinate (Q1) at (-2.5cm,4.5cm);
\coordinate (L1) at (-3.5cm,3.5cm);
\coordinate (R1) at (-1.5cm,3.5cm);

\draw[thick] (C1) -- (Q1); 	 
\draw[thick] (C2) -- (Q2); 

\draw[fill=black] (C1) circle (0.15cm);
\node [below right] at (C1) {$c_1$};
\draw[fill=black] (C2) circle (0.15cm);
\node [below right] at (C2) {$c_2$};
	
\draw[fill=blue] (Q1) circle (0.15cm);
\node [above right] at (Q1) {$r_1$};
\draw[fill=blue] (Q2) circle (0.15cm);
\node [above right] at (Q2) {$r_2$};

\begin{scope}[shift={($(-2,0)$)}]       
    \draw[->] (-3,0) -- (-2,0);
    \node [right] at (-2,0) {$x$};
    \draw[->] (-3,0) -- (-3,1);
    \node [left] at (-3,1) {$t$};
\end{scope}

\draw[dashed,blue] (-4,1) -- (4,1) -- (4,6) -- (-4,6) -- (-4,1);
\draw[dashed,red] (3.5,1.5) -- (3.5,5.5) -- (1.5,5.5) -- (1.5,1.5) -- (3.5,1.5);

\end{tikzpicture}
\caption{}
\end{subfigure}%
\hfill
\begin{subfigure}{.5\textwidth}
  \centering
  \begin{tikzpicture}[scale=0.60]

\coordinate (P) at (0cm, 0cm);

\coordinate (C2) at (2.5cm, 2.5cm);
\coordinate (Q2) at (2.5cm, 4.5cm);

\coordinate (C1) at (-2.5cm,-1.5cm);
\coordinate (Q1) at (-2.5cm,0.5cm);

\draw[thick] (C1) -- (Q1); 	 
\draw[thick] (C2) -- (Q2); 

\draw[fill=black] (C1) circle (0.15cm);
\node [below right] at (C1) {$c_1$};
\draw[fill=black] (C2) circle (0.15cm);
\node [below right] at (C2) {$c_2$};
	
\draw[fill=blue] (Q1) circle (0.15cm);
\node [above right] at (Q1) {$r_1$};
\draw[fill=blue] (Q2) circle (0.15cm);
\node [above right] at (Q2) {$r_2$};

\draw[dashed,red] (-4,-2.5) -- (4,-2.5) -- (4,6) -- (-4,6) -- (-4,-2.5);
\draw[dashed,blue] (3.5,1.5) -- (3.5,5.5) -- (1.5,5.5) -- (1.5,1.5) -- (3.5,1.5);

\draw[red][-triangle 45] (C1) -- (Q2);
    
\end{tikzpicture}
\caption{}
\end{subfigure}
\caption{Illustration of condition (\ref{cond:AUpairingcondition}) in theorem \ref{thm:authorizedunauthorizedassembly}. Dashed red boxes enclose unauthorized sets while dashed blue boxes enclose authorized sets. The condition states that every pairing $(\mathscr{A}_a,\mathscr{U}_i)$ of authorized with unauthorized set must have either $(a)$ $\mathscr{A}_a \setminus \mathscr{U}_i \neq \emptyset$ or $(b)$ $\mathscr{U}_i \setminus \mathscr{A}_a$ is causally connected to $\mathscr{A}_a$.}
\label{fig:condition3}
\end{figure}
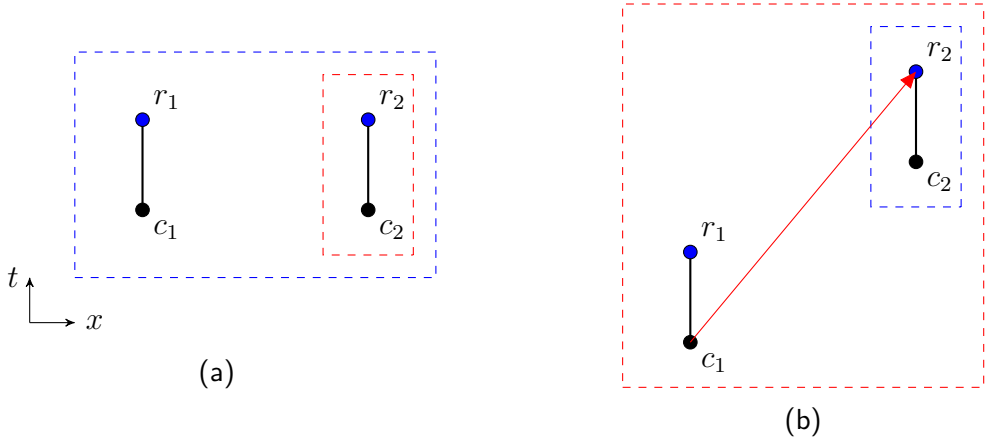

We can now proceed to characterize the state-assembly tasks with both authorized and unauthorized sets that can be completed by Alice. The difficulty here for Alice is different than in the case of localize-exclude. In localize-exclude, she had to keep the system out of a region $\mathscr{U}_i$ from an attacker who might gain full access to $\mathscr{U}_i$. Now, Alice's labs are secure. However the sets of spacetime points corresponding to an authorized call can be overlapping with those corresponding to an unauthorized call. This means that locally she may not be able to tell an authorized and unauthorized call apart. 

To understand under what conditions Alice can avoid an accidental reveal of the system to an unauthorized set of diamonds, we can first consider a task of the form $(A,s,\mathscr{A},\mathscr{U})$ having one authorized and one unauthorized set of diamonds. In this case, Alice can be successful if either (a) there is a diamond in $\mathscr{A}$ which is not in $\mathscr{U}$, since then she can turn over the system at that diamond only when there is a call there or (b) there is a diamond $D_*$ in $\mathscr{A}$ which, although it is in $\mathscr{U}$, is positioned such that Alice can tell at $D_*$ whether the global set of calls is authorized or unauthorized. In particular, this occurs exactly when there is a diamond in $\mathscr{U}\setminus \mathscr{A}$ which is causally connected to $\mathscr{A}$. Figure \ref{fig:condition3} illustrates these two possibilities. 

We now state and prove the theorem characterizing the state-assembly tasks with authorized and unauthorized sets of diamonds.
\begin{theorem}\label{thm:authorizedunauthorizedassembly}
A state-assembly task with authorized sets $\{\mathscr{A}_a\}$ and unauthorized sets $\{\mathscr{U}_i\}$ and start point $s$ can be completed if and only if the following three conditions hold:
\begin{enumerate}
\item The return point of at least one diamond from each authorized set is in the causal future of the start point.
\item Each pair of authorized sets $(\mathscr{A}_a,\mathscr{A}_b)$ is causally connected.
\item \label{cond:AUpairingcondition} Each pair $(\mathscr{A}_a, \mathscr{U}_i)$ of authorized with unauthorized sets has the property that either $\mathscr{A}_a \setminus \mathscr{U}_i \neq \emptyset$ or $\mathscr{U}_i \setminus \mathscr{A}_a$ is causally connected to $\mathscr{A}_a$.
\end{enumerate}
\end{theorem}
\begin{proof}\,The necessity of conditions (i) and (ii) follow from the same arguments as in theorem \ref{thm:authorizedassembly}. To see the necessity of condition (iii), consider that its negation is that both $\mathscr{A}_a \subset \mathscr{U}_i$ and $\mathscr{U}_i \setminus \mathscr{A}_a$ is not causally connected to $\mathscr{A}_a$. Then Alice's agents in the diamonds of $\mathscr{A}_a$, should they receive calls, cannot distinguish a call to $\mathscr{A}_a$ from a call to $\mathscr{U}_i$ since they are causally disconnected from diamonds in $\mathscr{U}_i\setminus \mathscr{A}_a$. In order to complete the task, Alice must always hand the system over to $\mathscr{A}_a$ when she receives a call there. She will then also always hand over the system when the call is to $\mathscr{U}_i$, leading to her failing the task. 

To demonstrate sufficiency we construct an explicit protocol to complete the task in the case where all three conditions are true. Using the error-correcting code constructed from the graph of causal connections (also used in theorems \ref{thm:localize}, \ref{thm:localexclude}, and \ref{thm:authorizedassembly}) we can reduce the initial $(A,s,\{\mathscr{A}_1,...,\mathscr{A}_n\},\{\mathscr{U}_1,...,\mathscr{U}_m\})$ task to many tasks of the form $(S_{ij},s,\{\mathscr{A}_i,\mathscr{A}_j\},\{\mathscr{U}_1,...,\mathscr{U}_m\})$, where the $S_{ij}$ are the shares of the error-correcting code associated to the $i-j$ pair of regions.

To complete the $(S_{ij},s,\{\mathscr{A}_i,\mathscr{A}_j\},\{\mathscr{U}_1,...,\mathscr{U}_m\})$ tasks, we encode the share $S_{ij}$ using the quantum one-time pad with some classical randomness $k_{ij}$. Now notice that we can complete the $(S_{ij},s,\{\mathscr{A}_i,\mathscr{A}_j\},\{\mathscr{U}_1,...,\mathscr{U}_m\})$ task by completing $(S_{ij},s,\{\mathscr{A}_i,\mathscr{A}_j\},\emptyset)$ on the quantum share $S_{ij}$ and $(k_{ij},s,\mathscr{A}_i,\{\mathscr{U}_1,...,\mathscr{U}_m\})$ and \\ $(k_{ij},s,\mathscr{A}_j,\{\mathscr{U}_1,...,\mathscr{U}_m\})$  on the classical string. Stated another way, the use of the one-time pad lets us ignore avoiding the unauthorized sets when considering the quantum data, and only worry about not handing the classical string over at the unauthorized sets. 

To complete the tasks of the form $(k_{ij},s,\mathscr{A}_i,\{\mathscr{U}_1,...,\mathscr{U}_m\})$, we encode $k_{ij}$ into a $((m,m))$ classical secret sharing scheme with shares labelled $k_{ij}^l$. Then completing the tasks $(k_{ij}^l,s,\mathscr{A}_i,\mathscr{U}_l)$ ensures each share $k_{ij}^l$ ends up at $\mathscr{A}_i$, so the string $k_{ij}$ can be constructed there along with the quantum share $S_{ij}$. At the same time, completing the tasks $(k_{ij}^l,s,\mathscr{A}_i,\mathscr{U}_l)$ ensures the share $k_{ij}^l$ is missing from $\mathscr{U}_l$, and since every share $k_{ij}^l$ is needed to recover $S_{ij}$, $S_{ij}$ cannot be decoded there. 

To complete these $(k,s,\mathscr{A}_i,\mathscr{U}_j)$ tasks, recall that by condition (iii) either $\mathscr{A}_i \setminus \mathscr{U}_j$ is not empty or $\mathscr{U}_i \setminus \mathscr{A}_i$ is causally connected to $\mathscr{A}_i$. If $\mathscr{A}_i \setminus \mathscr{U}_i$ is not empty, then send $k$ to a diamond $D_*$ in $\mathscr{A}_i \setminus \mathscr{U}_i$. Then hand over $k$ at $D_*$ if there is a call there. If $\mathscr{U}_i \setminus \mathscr{A}_i$ is causally connected to $\mathscr{A}_i$ then send $k$ to any diamond $D_*$ in $\mathscr{A}_i$ which has at least one call point of $\mathscr{U}_i \setminus \mathscr{A}_i$ in its causal past. Then hand over $k$ at $D_*$ if there is a call there and no call at the diamonds in $\mathscr{U}_j\setminus \mathscr{A}_i$.
\end{proof}

We give a task on four diamonds in figure \ref{fig:4diamondexample} and demonstrate how to complete it using the protocol constructed in this proof. 

\begin{figure}
\begin{center}
\begin{tikzpicture}[scale=1]

\draw[ultra thick] (0,0,0) -- (0,3,0);
\draw[fill=black] (0,0,0) circle (0.15);
\draw[fill=blue] (0,3,0) circle (0.15);

\draw[ultra thick] (0,0,3) -- (0,3,3);
\draw[fill=black] (0,0,3) circle (0.15);
\draw[fill=blue] (0,3,3) circle (0.15);

\draw[-triangle 45][red] (0,0,0) -- (0,3,3);
\draw[-triangle 45][red] (0,0,3) -- (0,3,0);

\node [below right] at (0,0,0) {$c_{a,1}$};
\node [below right] at (0,0,3) {$c_{a,2}$};

\node [above right] at (0,3,0) {$r_{a,1}$};
\node [above right] at (0,3,3) {$r_{a,2}$};

\draw[ultra thick] (5,0,0) -- (5,3,0);
\draw[fill=black] (5,0,0) circle (0.15);
\draw[fill=blue] (5,3,0) circle (0.15);

\draw[ultra thick] (5,0,3) -- (5,3,3);
\draw[fill=black] (5,0,3) circle (0.15);
\draw[fill=blue] (5,3,3) circle (0.15);

\draw[-triangle 45][red] (5,0,0) -- (5,3,3);
\draw[-triangle 45][red] (5,0,3) -- (5,3,0);

\node [below right] at (5,0,0) {$c_{b,1}$};
\node [below right] at (5,0,3) {$c_{b,2}$};

\node [above right] at (5,3,0) {$r_{b,1}$};
\node [above right] at (5,3,3) {$r_{b,2}$};

\begin{scope}[shift={($(0,-1,0)$)}]       
    \draw[->] (-3,0,1) -- (-3,0,-0.2);
    \node [right] at (-3,0,-0.2) {$y$};
    \draw[->] (-3,0,1) -- (-2,0,1);
    \node [right] at (-2,0,1) {$x$};
    \draw[->] (-3,0,1) -- (-3,1,1);
    \node [right] at (-3,1,1) {$t$};
\end{scope}
	
\end{tikzpicture}
\end{center}
\caption{An arrangement of causal diamonds on which we can define an assembly task. Define authorized sets $\mathscr{A}_1=\{D_{a,1},D_{b,1}\}$ and $\mathscr{A}_2=\{D_{a,2},D_{b,2}\}$ while any set of three or four diamonds is deemed unauthorized. One can check that every unauthorized set has $\mathscr{U}_i\setminus \mathscr{A}_j$ causally connected to $\mathscr{A}_j$, so theorem \ref{thm:authorizedunauthorizedassembly} gives that this task can be completed. To do so, the initial system $A$ is encoded using the quantum one-time pad and sent towards the pair of diamonds labelled `$a$', $D_{a1}$ and $D_{a2}$. It should be handed over at whichever diamond receives a call. The key $k$ from the one-time pad is stored in a $((2,2))$ secret sharing scheme as $k=k_a\oplus k_b$ and $k_a$ and $k_b$ are sent towards the `$a$' and `$b$' pairs of diamonds respectively. At the `$a$' pair of diamonds, $k_a$ is returned to $D_{a1}$ if there is a call there and no call at $D_{a2}$, or at $D_{a2}$ if there is a call there and no call at $D_{a1}$. The $k_b$ string is returned to $D_{b1}$ or $D_{b2}$ using the same logic. If three or four diamonds receive calls, then at least one of the `$a$' or `$b$' pairs of diamonds will not receive a share of the $((2,2))$ scheme, so Bob will not receive the $A$ system. Notice that the task is possible even though $\mathscr{A}_1,\mathscr{A}_2$ are subsets of the unauthorized sets, violating the monotonicity requirement of quantum secret sharing.}
\label{fig:4diamondexample}
\end{figure}
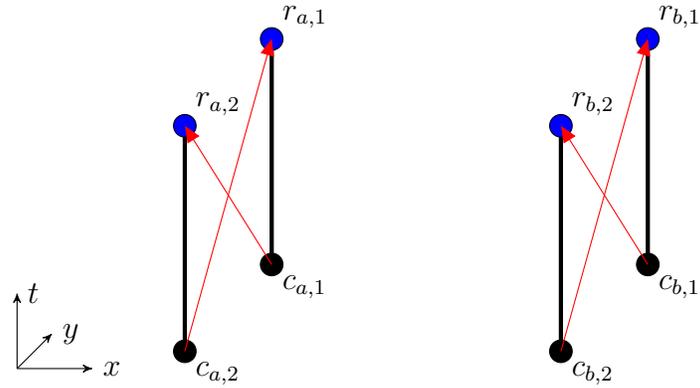

The state-assembly task seems a less natural extension of quantum secret sharing to spacetime, since condition (iii) differs notably from the corresponding condition in secret sharing. In particular, some allowed state-assembly tasks have unauthorized sets which contain authorized ones, violating the monotonicity requirement of quantum secret sharing~\cite{gottesman2000theory}. In contrast, the localize-exclude task mimics the monotonicity requirement closely, since the condition there is that (the domain of dependence of) the unauthorized region not contain the authorized region. However, this distinction from secret sharing opens up interesting new possibilities; in the next section we propose a cryptographic task and protocol which exploits the failure of monotonicity in the state-assembly task. 

\section{An application: party-independent transfer}\label{sec:partyindependenttransfer}

As discussed in the introduction, relativistic tasks in Minkowski space have provided an interesting set of tools for the cryptographer. In part, our motivation for considering the state-assembly task with authorized and unauthorized regions is in the hope it will find such application. The state-assembly task includes scenarios with many parties, and allows for a rich array of possible causal structures. Each causal structure translates to a set of restrictions on which parties can know what, and when, and it seems plausible that these restrictions can be exploited to perform some interesting multiparty task or computation securely.

We suspect there are many possible directions to consider, and make a small start at this by suggesting below one particular task. We do not offer complete security arguments for our proposal or careful discussion of the practicality of this task. Our aim is simply to suggest the applicability of the state-assembly task to cryptography.

To motivate the task consider the following scenario. Alice is an employer who wishes to hire either Bob$_1$ or Bob$_2$. Alice is known to be inclined to prejudice, and the Bobs wish to ensure they are paid based on the work done alone, without regard to their identity. An easy solution would be to announce publicly the position's salary, but unfortunately the Bobs are private people. They wish to keep their salaries secret while also having a guarantee of fairness. We define the \emph{party-independent transfer} task in order to satisfy these two competing needs. 

In the party-independent transfer task, we specify that each Bob will give an input $X_i$ to Alice. Alice will then output quantum systems $S(X_1,X_2,a)$ and $T(X_1,X_2,a)$, with one system handed to each of the Bobs, where $a$ is a variable fixed by Alice. The task occurs in a spacetime setting, so in general the $X_i$ may be stored as several bits handed over from Bob's agents to Alice's agents at distributed spacetime points. The $X_i$ should be distinct. If not, then the protocol aborts.

To meet the needs of our jealousy-prone but private Bobs, and guard against the prejudiced Alice, we need the transfer to have the following properties:
\begin{enumerate*}
\item \emph{Party independence:} The output systems $S$ and $T$ produced by Alice have the property that 
\begin{eqnarray} \label{eq:indep}
S(X_1,X_2,a) = T(X_2,X_1,a) \,\,\,\,\, and \,\,\,\,\, S(X_2,X_1,a) = T(X_1,X_2,a).
\end{eqnarray}
In words, we require that the output given to Bob$_1$ would have been given to Bob$_2$ had the Bobs reversed their inputs. 
\item \emph{Fixed:} As a set, $\{S(X_1,X_2,a),T(X_1,X_2,a)\}$ is determined by the variable $a$ only. In words, the Bobs' input influences who receives which system only, not which two systems are handed over.
\item \emph{Secret:} Each Bob does not learn Alice's output to the other Bob. In particular, this requires that Alice not satisfy condition 1 trivially by having $S(X_1,X_2,a)=T(X_1,X_2,a)$ always.
\end{enumerate*}

To assure ourselves completing this task is not trivial consider various naive approaches. We might have Alice share two entangled sets of degrees of freedom, $E_1$ given to Bob$_1$ and $E_2$ given to Bob$_2$, onto which she will later teleport $T$ and $S$, respectively. The Bobs could then exchange degrees of freedom if they decide to reverse the arrangement of who receives which system. This is certainly party-independent, since Alice performs the teleportation without knowing who holds which degrees of freedom. However, the fixed property is violated, as either Bob can act on their degrees of freedom before exchanging it. 

Another strategy would be to have Alice publicly announce a protocol for preparing each of $S$ and $T$. Clearly this is fixed and party-independent, but fails to be secret. Finally, Alice could separately hand $S$ to Bob$_1$ and $T$ to Bob$_2$ (or vice versa). This would be fixed and secret but not party-independent.

Although the obvious strategies fail, the state-assembly task seems to be well-suited to achieving party-independent transfer. As intuition, we can note that in a state-assembly task Alice's agents, who only have access to local information and not the global set of calls made by the Bobs, may not be aware of who has received the system until a late time when she has been able to collect and compare all of the call data. Further, we have already introduced the notion of an unauthorized set of calls and can hope to exploit this to achieve the secrecy property of party-independent transfer. 

\begin{figure}
\begin{center}
\begin{tikzpicture}[scale=1]

\draw[ultra thick] (0,0,0) -- (0,3,0);
\draw[fill=black] (0,0,0) circle (0.15);
\draw[fill=blue] (0,3,0) circle (0.15);

\draw[ultra thick] (0,0,3) -- (0,3,3);
\draw[fill=black] (0,0,3) circle (0.15);
\draw[fill=blue] (0,3,3) circle (0.15);

\draw[-triangle 45][red] (0,0,0) -- (0,3,3);
\draw[-triangle 45][red] (0,0,3) -- (0,3,0);

\node [below right] at (0,0,0) {$c_{a,1}$};
\node [below right] at (0,0,3) {$c_{a,2}$};

\node [above right] at (0,3,0) {$r_{a,1}$};
\node [above right] at (0,3,3) {$r_{a,2}$};

\draw[ultra thick] (5,0,0) -- (5,3,0);
\draw[fill=black] (5,0,0) circle (0.15);
\draw[fill=blue] (5,3,0) circle (0.15);

\draw[ultra thick] (5,0,3) -- (5,3,3);
\draw[fill=black] (5,0,3) circle (0.15);
\draw[fill=blue] (5,3,3) circle (0.15);

\draw[-triangle 45][red] (5,0,0) -- (5,3,3);
\draw[-triangle 45][red] (5,0,3) -- (5,3,0);

\node [below right] at (5,0,0) {$c_{b,1}$};
\node [below right] at (5,0,3) {$c_{b,2}$};

\node [above right] at (5,3,0) {$r_{b,1}$};
\node [above right] at (5,3,3) {$r_{b,2}$};

\draw[ultra thick] (10,0,0) -- (10,3,0);
\draw[fill=black] (10,0,0) circle (0.15);
\draw[fill=blue] (10,3,0) circle (0.15);

\draw[ultra thick] (10,0,3) -- (10,3,3);
\draw[fill=black] (10,0,3) circle (0.15);
\draw[fill=blue] (10,3,3) circle (0.15);

\draw[-triangle 45][red] (10,0,0) -- (10,3,3);
\draw[-triangle 45][red] (10,0,3) -- (10,3,0);

\node [below right] at (10,0,0) {$c_{c,1}$};
\node [below right] at (10,0,3) {$c_{c,2}$};

\node [above right] at (10,3,0) {$r_{c,1}$};
\node [above right] at (10,3,3) {$r_{c,2}$};

\begin{scope}[shift={($(0,-1,0)$)}]       
    \draw[->] (-3,0,1) -- (-3,0,-0.2);
    \node [right] at (-3,0,-0.2) {$y$};
    \draw[->] (-3,0,1) -- (-2,0,1);
    \node [right] at (-2,0,1) {$x$};
    \draw[->] (-3,0,1) -- (-3,1,1);
    \node [right] at (-3,1,1) {$t$};
\end{scope}
	
\end{tikzpicture}
\end{center}
\caption{Arrangement of call-reveal pairs used in the proposed party-independent transfer protocol. Bob$_1$ controls the diamonds $D_{a,1},D_{b,1},D_{c,1}$ while Bob$_2$ controls $D_{a,2},D_{b,2},D_{c,2}$.} \label{fig:partyindependenttransfer}
\end{figure}
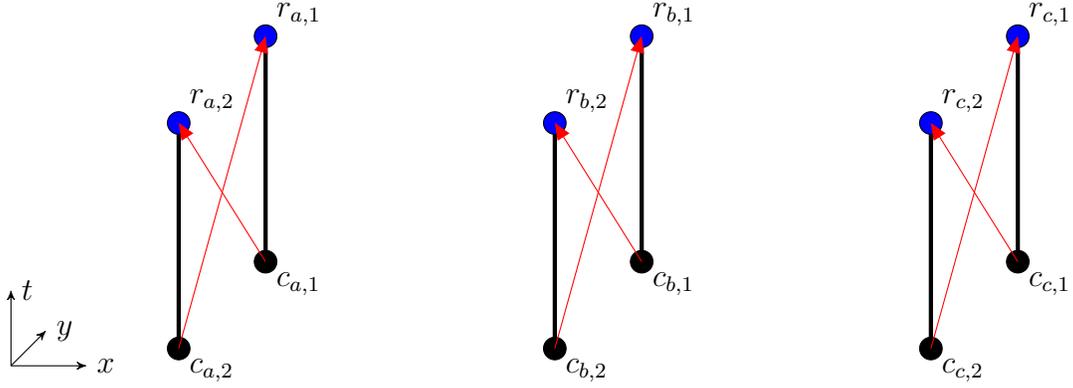

Indeed, we can put forward a candidate protocol built on a state-assembly task that seems to achieve all three security requirements of party-independent transfer. Before explaining the protocol however, we need to highlight one feature of the $((2,3))$ secret sharing scheme which will be used. We will use an error-correcting code on three physical qutrits which stores one logical qutrit. The logical states are given by
\begin{align}\label{eq:codesubspace}
\ket{0_L} &= \frac{1}{\sqrt{3}}(\ket{000}+\ket{111}+\ket{222}), \nonumber \\
\ket{1_L} &= \frac{1}{\sqrt{3}}(\ket{012}+\ket{201}+\ket{120}), \nonumber \\
\ket{2_L} &= \frac{1}{\sqrt{3}}(\ket{021}+\ket{102}+\ket{210}).
\end{align}
One may check explicitly that there exists a decoding operation $U_{12}^\dagger$ supported on the first two qutrits such that 
\begin{align}
\mathcal{U}_{12} ^\dagger \ket{i_L} = \ket{i}_1 \ket{\chi}_{23},
\end{align}
where
\begin{align}\label{eq:chistate}
\ket{\chi} = \frac{1}{\sqrt{3}} \left(\ket{00}+\ket{11}+\ket{22} \right).
\end{align}
By the symmetry in the code, a similar decoding operation exists for any subsystems of two qutrits. We wish to highlight that after the decoding operation is applied, two of the qutrits are left in a maximally entangled state. 

To construct the protocol, we will use the arrangement of diamonds shown in figure \ref{fig:partyindependenttransfer}. Bob$_1$ controls the diamonds $D_{a,1},D_{b,1},D_{c,1}$ while Bob$_2$ controls $D_{a,2},D_{b,2},D_{c,2}$. We consider a scenario where the Bobs choose at random which of them receives which system, although modifications to this are easy. We divide the protocol into a preparation phase, transfer phase, and checking phase for clarity in presentation.
\begin{protocol}\label{protocol:pit}
\textbf{Compensation protocol}
\end{protocol}
\begin{enumerate*}
\item \textbf{Preparation phase}
\begin{enumerate*}
\item Alice prepares a quantum state $\ket{\Psi}_{SR}$, and encodes the $S$ system into the $((2,3))$ secret sharing scheme using the encoding given in equation \ref{eq:codesubspace}.
\item Bob$_1$ and Bob$_2$ execute a coin flipping protocol. The outcome is not revealed to Alice. Without loss of generality, suppose that Bob$_1$ wins the coin toss, which determines that he should receive $S$.
\item Bob$_1$ chooses at random two of the three diamonds he controls and sends calls to each of them. Without loss of generality, we call these diamonds $D_{a,1}$ and $D_{b,1}$. Bob$_2$ then sends a call to the diamond he controls which is not causally connected to $D_{a,1}$ or $D_{b,1}$, which in this case is $D_{c,2}$.
\end{enumerate*}
\item \textbf{Transfer phase}
\begin{enumerate*}
\item Alice routes one share of her secret sharing scheme towards each of the diamond pairs labelled by $a,b$ and $c$.
\item Alice responds to the summons at each of the call points by comparing the calls from $D_{x,1}$ and $D_{x,2}$. If both have $b=1$ or both have $b=0$, Alice does not hand over the share to either diamond. If exactly one of the two diamonds has $b=1$, Alice hands the share over at the corresponding return point.
\end{enumerate*}
\item \textbf{Checking phase}
\begin{enumerate*}
\item Bob$_1$ applies the decoding map to his two shares, producing $\ket{\Psi}_{SR} \otimes \ket{\chi}$ where $\ket{\chi}$ is the maximally entangled state given in equation \ref{eq:chistate}. 
\item Bob$_2$ sends his share of the maximally entangled state to Bob$_1$, who then measures the pair jointly to ensure he holds $\ket{\chi}$.
\end{enumerate*}
\end{enumerate*}
In the notation of our security definition, the inputs $X_i$ by the Bobs consist of their three output bits $X_i=\{b_{i,a},b_{i,b},b_{i,c}\}$. Before the checking phase, the state is $\ket{\Psi}_{SR}\otimes \ket{\chi}$, with the receiving Bob holding $S$ and half the maximally entangled state $\ket{\chi}$, and the non-receiving Bob holding the other half of $\ket{\chi}$. After the checking phase however one Bob holds only $S$ and the $\ket{\chi}$ state has been measured. We should then identify the $T$ system of the definition as $T = \emptyset$. If we would like both Bobs to receive some quantum system we can run the protocol twice. 

We can argue for the secret and fixed properties of this protocol. Fixed is clear, since the receiving Bob can reconstruct the system from degrees of freedom that have never been held by the non-receiving Bob. Regarding secrecy, we note that the non-receiving Bob receives only one share of the secret sharing scheme, so learns no information about $S$. The non-receiving Bob may try to receive additional shares by sending additional calls, but in this case Alice will notice that calls have been made at two causally related diamonds and not hand over any shares to those diamonds. 

To argue for party-independence, note that Alice is already limited in her knowledge of who is receiving the system. Although at each pair of diamonds she knows whether she is handing a single share over to Bob$_1$ or Bob$_2$, none of Alice's agents have the global information of which Bob is receiving two shares, and thus the system $S$. Later on she will be able to collect information from all the call points and determine this, but at the spacetime points of transfer this is not known. Alice might try to have one set of shares which she hands to Bob$_1$ and a separate set to Bob$_2$, but using two unentangled sets of shares for Bob$_1$ and Bob$_2$ will lead to a failure in the checking phase. We leave proving or disproving the security of this protocol, which we regard as plausible but not obvious, to future work.

It is perhaps useful to note a connection of classical bit commitment with the party-independent transfer task. Given a bit commitment scheme which consists of 1) Alice handing a commitment to Bob, then 2) Alice later handing a reveal to Bob, which he uses to access Alice's committed bit, it is possible to construct a party-independent transfer protocol.\footnote{This was pointed out to the authors on the cryptography stack exchange \cite{stack}.} In particular, Alice publicly announces her commitment to both Bob$_1$ and Bob$_2$, then hands the reveal to only one of the Bobs. However, it is known that there are no unconditionally secure bit commitment schemes of this form \cite{lo1997quantum,mayers1997unconditionally}.

\section{Discussion}

In our first article on summoning \cite{hayden2016summoning}, we argued that the summoning task gives an operational setting in which to understand how quantum information can and cannot move through spacetime. That setting was restricted however to asking if a quantum system could be localized to collections of causal diamonds. 

In this article we have generalized in a way that allows us to ask if a quantum system is localized to a collection of arbitrary spacetime regions. We have defined the notion of localized by allowing some party with no prior knowledge of the system unrestricted access to the spacetime region. If they can later construct the system then we say it was localized there; if they learn nothing about the system we say it is excluded. This is consistent with our previous definition of localization to a diamond, in that completing the summoning task means in particular that the system was localized to each diamond. However, the notion of localization implied by summoning is stronger than the notion used in this article, since in summoning Alice must perform the data processing needed to construct the system while within the diamond. In the localize task this data processing can occur outside the region.

In the absence of gravity, where there are no known limits on the rate of computation, the strong and weak notions of localization coincide, at least for diamond-shaped regions. In the presence of gravity Lloyd argued there is a limit on computational speed \cite{lloyd2000ultimate} but there are counterexamples to his proposed bound~\cite{jordan2017fast}. It is nonetheless plausible that computational speed is limited by quantum gravity, so one can imagine a scenario where a quantum system is localized to a region in the weaker sense (in that it is possible to construct it from systems that pass through that region) but not in the strong sense (in that it is impossible to do so within the spatial-temporal extent of that region due to gravitational constraints on computation). Thus, in the presence of gravity these notions of localization plausibly become distinct. Attempts to resolve fundamental puzzles like the black hole information paradox~\cite{susskind1993stretched,hayden2007black,harlow2013quantum} have also hinted at this distinction, and indicate that it may be the stronger notion of localization for which the no-cloning theorem applies.

Also in the context of gravity the holographic bound \cite{bousso2002holographic} makes tasks with sufficiently small regions or sufficiently large numbers of regions impossible to complete, since it places a limit on how many qubits may be localized to a region of a given area without producing a black hole. Thus, we should understand the theorems given in this work as applying only in the absence of gravity. It would be interesting to perform a detailed study of exceptions to our theorems arising from gravitational physics.

By adding excluded regions to the localize task we have found a natural extension of quantum secret sharing to a spacetime setting. Indeed, the conditions for completing the localize-exclude task have close analogues in the conditions for constructing quantum secret sharing schemes, and we can embed any quantum secret sharing scheme as a carefully chosen localize-exclude task. The conditions on the start point in the localize-exclude task are somewhat awkward from this perspective, but can be seen as corresponding to certain trivial requirements in the secret sharing language. 

Since the localize-exclude task corresponds so closely to quantum secret sharing, we might expect that it doesn't provide any new tools for the construction of cryptographic protocols. From this perspective the state-assembly task is more interesting, since there we can have an unauthorized set contain an authorized one. This violates the monotonicity requirement that occurs in both localize-exclude and quantum secret sharing. 

We have given one proposed application that exploits this violation of monotonicity: party-independent transfer. This proposal is in need of a more complete study. We have not proven our proposed protocol is secure, nor considered what more practical goals within cryptography this primitive may be used to achieve. It would also be interesting to understand the relation of the proposed party-independent transfer task to established cryptographic primitives. We have already pointed out a connection to bit commitment, but there may also be interesting relations to (for instance) the spacetime analogues of oblivious transfer mentioned in the introduction. 

\section{Acknowledgements}

This work was started while the authors were at McGill University, and restarted while attending the first It from Qubit summer school at the Perimeter Institute for Theoretical Physics. This work also benefited from the Quantum Physics of Information program held at the Kavli Institute for Theoretical Physics in Santa Barbara, and from a visit by AM to the Stanford Institute for Theoretical Physics. 

AM wishes to acknowledge the UBC REX program and his mentees Andrew Chun and Liam Vanderpoel, discussions with whom motivated some of the results given here. The authors would also like to thank Adrian Kent for many helpful discussions. Daniel Gottesman pointed out to us the use for the quantum one-time pad given here as figure \ref{fig:onetimepad}. We are also grateful for discussions with Eric Hanson, Fang Xi Lin, David Stephen, Sepehr Nezami, Geoff Penington, Grant Salton and Leonard Susskind. Kevin Milner, David Wakeham, and Jason Pollack provided useful feedback on this manuscript.

Research at Perimeter Institute is supported by the Government of Canada through Industry Canada and by the Province of Ontario through the Ministry of Economic Development \& Innovation. This research was supported in part by the National Science Foundation under Grant No. NSF PHY17-48958. AM was supported by a NSERC C-GSM award, later a NSERC C-GSD award, and by the It from Qubit collaboration sponsored by the Simons Foundation. PH was supported by AFOSR (FA9550-16-1-0082), CIFAR, and the Simons Foundation. 

\appendix

\section{Summoning, state-assembly and localization} \label{appendix:equivalence}

In the main article we have discussed two related tasks: state-assembly and localize-exclude. A third task, summoning, has also been considered in earlier work \cite{hayden2016summoning}. All these tasks relate to how quantum information can move through spacetime; in this appendix we clarify the relationships among these three tasks. 

The summoning task was introduced by Kent \cite{kent2013no} and expanded on later in \cite{hayden2016summoning}. There are various variations on its definitions, as we discuss below. We give the definition from \cite{hayden2016summoning} first. 
\begin{definition}
A \textbf{single-call single-return summoning task} is a task involving two interacting agencies, Alice and Bob. The task is defined by 
\begin{enumerate}
    \item A quantum system $A$, where Bob knows the state of the purification $\ket{\Psi}_{AR}$ and holds the $R$ system.
    \item A start point $s$ at which Bob gives Alice $A$
    \item A collection of causal diamonds $D_i$, each of which is defined by a call point $c_i$ and return point $r_i$
\end{enumerate}
At each call point $c_i$ Alice receives a classical bit $b_i$. Alice is guaranteed that exactly one bit will be $1$, say $b_{i^*}=1$, and the remainder will have $b_j=0$, but does not know the value of $i^*$ in advance. To successfully complete the task, Alice should return the system $A$ to the point $r_{i^*}$ such that $b_{i^*}=1$.
\end{definition}
To complete the summoning task Alice must send systems sufficient to reconstruct the system through each diamond. Consequently, completing a summoning task with diamonds $\{D_i\}$ also completes an associated localize task with authorized regions $\mathscr{A}_i=D_i$. However, the reverse is not true: completing the localize task implies some collection of systems inside each authorized region can be used to construct the system, but doesn't require that this reconstruction can take place within the region. For instance, exhibiting the system could require the application of a high complexity circuit, perhaps requiring so many gates that gravitational speed limits would prevent their completion in the required time. Further, localize tasks deal with regions of arbitrary shape, whereas in a summoning task only causal diamond shaped regions are discussed.

The basic restriction on when a summoning task may be completed is the no-summoning theorem \cite{kent2013no}.
\begin{theorem}
A single-call single-return summoning task with two diamonds $D_1,D_2$ is impossible whenever $D_1$ and $D_2$ are causally disjoint.
\end{theorem}
\begin{proof}\,Suppose there exists a protocol that returns the system to $r_1$ when there is a call to $c_1$ and returns the system to $r_2$ when there is a call to $c_2$. Then we can argue such a protocol can be used to clone a quantum system, and consequently no such protocol can exist. To see this, suppose there is a call to both $c_1$ and $c_2$. Then since $D_1$ and $D_2$ are causally disjoint, Alice's agent at $D_1$ cannot distinguish this case from the case where $c_1$ receives a call and $c_2$ does not. By assumption then she returns the system to $r_1$. Similarly, Alice's agent at $D_2$ returns the system at $r_2$. Alice has then handed over two copies of the quantum system. 
\end{proof}

The proof of the no-summoning theorem is similar to the proof of the no-assembly theorem we gave in the main text, see theorem \ref{thm:noassembly}. 

Similar to the localize task, summoning is possible whenever each pair of diamonds are causally connected and every diamond has a point in the future light cone of the start point. 
\begin{theorem}
The single-call single-return summoning task is possible if and only if:
\begin{enumerate*}
    \item The return point of each diamond is in the causal future of the start point.
    \item Every pair of diamonds  $(D_i,D_j)$ is causally connected. 
\end{enumerate*}
\end{theorem}
We omit the proof of this theorem as it proceeds along now familiar lines: the many diamond case is reduced to a two diamond case by use of an error-correcting code, which can be constructed from the graph of causal connections among the regions. In the case of two diamonds we complete the task using the teleportation protocol illustrated in figure \ref{fig:teleportation}.

The summoning task as given above is ``single-call'' in that exactly one of the $b_i=1$, and ``single-return'' in that the system should be returned in full at the called-to diamond. We can generalize this to allow for Alice to receive many calls (many $b_i=1$) in two possible ways. First, we might specify that Alice return a subsystem at each called-to diamond such that taken together these subsystems can be used to reproduce the $A$ system. In this case we have weakened the requirement on Alice --- she need not hand over the quantum system itself, just quantum information and classical instructions sufficient for Bob to later construct the system. We will refer to this as many-call many-return summoning. Alternatively, we can specify that Alice hand over the system itself at one (but any one) of the called-to diamonds. We call this many-call single-return summoning. This second case is treated by Adlam and Kent \cite{adlam2015quantum} and discussed further in appendix \ref{appendix:manycallsinglereturn}. The first case is closely related to the state-assembly task. We elaborate on this relation in the remainder of this section.

\begin{figure}[t]
\begin{center}
\begin{tikzpicture}
	\coordinate (PT) at (0,0,0);
	
	\coordinate (CL) at (1.5,1.5,-2);
	\coordinate (CR) at (1.5,0,2);
	
	\coordinate (QL) at (1.5,3,-2);
	\coordinate (QR) at (1.5,3,2);
	
	\draw[dashed] (1.5,3,2) -- (1.5,3,-2) -- (1.5,0,-2) -- (1.5,0,2);
	\draw[dashed] (1.5,0,-2) -- (0,0,-2) -- (0,0,2) -- (1.5,0,2);
	\draw[dashed] (1.5,3,-2) -- (0,3,-2) -- (0,3,2) -- (1.5,3,2);
	\draw[dashed] (0,3,-2) -- (0,0,-2);
	\draw[dashed] (0,3,2) -- (0,0,2);
	
	\draw[ultra thick] (CL) -- (QL);
	\draw[ultra thick] (CR) -- (QR);
	
	\draw[fill=yellow] (PT) circle (0.15);
	\node [above left] at (PT) {$s$};
	
	\draw[fill=black] (CL) circle (0.15);
	\draw[fill=black] (CR) circle (0.15);
	
	\node [below right] at (CL) {$c_1$};
	\node [below right] at (CR) {$c_2$};
	
	\draw[fill=blue]  (QL) circle (0.15);
	\draw[fill=blue]  (QR) circle (0.15);
	
	\node [above left] at (QL) {$r_1$};
	\node [above right] at (QR) {$r_2$};
    
    \begin{scope}[shift={($(0,0)$)}]       
    \draw[->] (-3,0,1) -- (-3,0,-0.2);
    \node [right] at (-3,0,-0.2) {$y$};
    \draw[->] (-3,0,1) -- (-2,0,1);
    \node [right] at (-2,0,1) {$x$};
    \draw[->] (-3,0,1) -- (-3,1,1);
    \node [right] at (-3,1,1) {$t$};
    \end{scope}
    
	\draw[red][-triangle 45] (PT) -- (QL);
	\draw[red][-triangle 45] (PT) -- (QR);
    \draw[red][-triangle 45] (CR) -- (QL);
    
\end{tikzpicture}
\end{center}
\caption{A summoning task on two diamonds in $2+1$ dimensions. In this task $r_1$ is in the future light cone of $c_2$, but $r_2$ is not in the future of $c_1$. (In all figures, red arrows indicate causal curves.) Additionally, $r_1$ and $r_2$ are in the future light cone of $s$. To complete this summoning task, Alice pre-shares entanglement between $s$ and $c_1$. At $s$, Alice teleports the $A$ system using the shared entanglement and then sends the classical teleportation data to both $r_1$ and $r_2$. At $c_1$, Alice routes the entangled particle she holds to $r_1$ if she receives $b_1=1$, and routes the particle to $r_2$ otherwise. This example is due to Kent \cite{kent2012quantum}.} \label{fig:teleportation}
\end{figure}
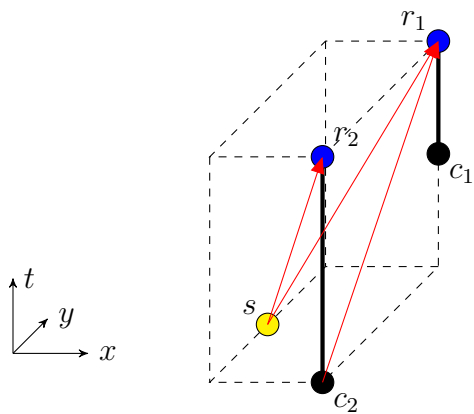

We can collect the discussion in the last paragraph into a definition of the many-call many-return summoning task.
\begin{definition}
A \textbf{many-call many-return summoning task} is a task involving two interacting agencies, Alice and Bob. A task is defined by 
\begin{enumerate}
    \item A quantum system $A$, where Bob knows the state of the purification $\ket{\Psi}_{AR}$ and holds the $R$ system
    \item A start point $s$ at which Bob gives Alice $A$
    \item A collection of authorized sets $\{\mathscr{A}_1,...,\mathscr{A}_n\}$ each consisting of one or more causal diamonds, $\mathscr{A}_i =\{D_{i1},...,D_{ik_i}\}$
\end{enumerate}
At the call point associated with each diamond Alice receives a bit $b_i$ from Bob. Alice has a guarantee that the calls will be to one of the authorized sets of diamonds. Alice is required to return a collection of classical and quantum systems at the associated $r_i$ which is sufficient to reconstruct $A$.
\end{definition}

We characterize the many-call many-return summoning tasks which are possible and those which are impossible in the following theorem.
\begin{theorem}\label{thm:manycallmanyreturn}
The many-call many-return summoning task is possible if and only if:
\begin{enumerate}
\item The return point of at least one diamond from each authorized set is in the causal future of the start point.
\item Every pair of authorized sets $(\mathscr{A}_i,\mathscr{A}_j)$ is causally connected. 
\end{enumerate}
\end{theorem}
Again we omit the proof, which follows the pattern of using the error-correcting code constructed from the graph of causal connections to reduce the many authorized set case to the two authorized set case. In the case of two authorized sets we use that the sets are causally connected, so in particular there exists a pair of causal diamonds chosen across the sets which are causally connected. We then complete the summoning task on these two diamonds using the teleportation protocol illustrated in figure \ref{fig:teleportation}.

From theorems \ref{thm:authorizedassembly} and \ref{thm:manycallmanyreturn} we find that state-assembly and summoning are possible for exactly the same arrangements of authorized sets. This is interesting, as although the tasks are similar they have one key distinction. In summoning Alice holds the $A$ system of an unknown quantum state $\ket{\Psi}_{AR}$, so can't produce copies of the $A$ system due to the linearity of quantum mechanics; in state-assembly Alice holds a known quantum state, but has the additional requirement that she hand over the $A$ system at most once. Thus, in the assembly task the requirement that Alice hand the system over at most once replaces the no-cloning restriction. The system Alice holds is essentially classical, since it is known to her and she may produce an arbitrary number of copies, but this gives her no additional power. In this sense we can view the state-assembly task as a classical analogue of the summoning task\footnote{Shortly before the publication of this manuscript reference \cite{kent2018unconstrained} appeared, which also discusses a classical version of the summoning task and its relation to the quantum one.}.

In the main article we discussed the state-assembly task with unauthorized regions. One could also consider a generalization of the summoning task with unauthorized regions, but this generalization is less well motivated. In particular, in the summoning task Bob both gives the system to Alice and requests it from her. It is unclear in what circumstance Alice would want to hide the system Bob gave to her from Bob when certain sets of calls are made. In the assembly setting this is more natural, since Alice has herself prepared the system and may want to hide it from certain subsets of other parties. 

\section{Many-call single-return summoning}\label{appendix:manycallsinglereturn}

In appendix \ref{appendix:equivalence} we discussed the many-call many-return summoning task, which we found is closely related to the state-assembly task discussed in the main article. Many-call many-return summoning is also interesting from the viewpoint of spacetime localization. In particular, completing the many-call many-return summoning task also completes the localize task. However, a second generalization of summoning to include many-calls is possible: we can consider a task with many calls but a single return, where Alice receives several calls from Bob and must return the system in full at exactly one (but any one) of the called-to diamonds.

We give a definition of the many-call single-return summoning task below.
\begin{definition}
A \textbf{many-call single-return summoning task} is a task involving two interacting agencies, Alice and Bob, defined by:
\begin{enumerate}
    \item A quantum system $A$, where Bob knows the state of the purification $\ket{\Psi}_{AR}$ and holds the $R$ system
    \item A start point $s$ at which Bob gives Alice system $A$
    \item A collection of authorized sets $\{\mathscr{A}_1,...,\mathscr{A}_n\}$ each consisting of one or more causal diamonds, $\mathscr{A}_i =\{D_{i1},...,D_{ik_i}\}$
\end{enumerate}
At the call point associated with each diamond Alice receives a bit $b_i$ from Bob. Alice has a guarantee that calls will be to one of the authorized sets. To successfully complete the task, Alice must return the $A$ system at exactly one of the called to diamonds.
\end{definition}
As defined here, the many-call single-return summoning task is somewhat more general than the task considered by Adlam and Kent. They considered in particular the case where the set of authorized sets $\{\mathscr{A}_1,...,\mathscr{A}_n\}$ corresponds to every possible subset of the diamonds. We refer to this as \textbf{unrestricted-call single-return} summoning.

Adlam and Kent characterized the full set of possible arrangements of diamonds for this unrestricted-call single-return summoning task~\cite{adlam2015quantum}. We recall their theorem here.
\begin{theorem}\label{thm:aksummoning}
\textbf{(Adlam and Kent 15')} Summoning with unrestricted calls with the requirement that Alice return the system at exactly one diamond is possible if and only if the following two conditions are true:
\begin{enumerate}
\item Every return point $r_i$ is in the future light cone of the start point $s$.
\item For any subset $\{D_{i_1},D_{i_2},...,D_{i_n}\}$ of diamonds, there is at least one diamond $D_{i_*}$ in the subset for which $r_{i_*}$ is in the future light cone of all the $c_i$ in the subset. \label{enum:many-2}
\end{enumerate}
\end{theorem}
Interestingly, condition (\ref{enum:many-2}) above is stronger than the corresponding condition for summoning with a single call. Adlam and Kent used this fact to argue against our interpretation of summoning in terms of localization of quantum information \cite{adlam2015quantum}; they argue that completing the summoning task depends on some resource provided to Alice by Bob --- a bit string of the form $000...010...000$ --- and thus that Alice is not localizing the system to each diamond. Instead, she is only successfully responding to the summons $b_i=1$ by exploiting her knowledge that certain other calls are $b_j=0$.

\begin{figure}
\begin{center}
\begin{tikzpicture}[scale=1]

	\coordinate (PS) at (2.5,-3.5,1.5);

	\coordinate (Ca) at (0,0,0);
	\coordinate (Cb) at (5,0,0);
	\coordinate (Cc) at (3,0,5);
	
	\coordinate (Qa) at (2.5,2.5,0);
	\coordinate (Qb) at (4, 2.5, 2.5);
	\coordinate (Qc) at (1.5, 2.5, 2.5);
	
	\draw[dashed] (Ca) -- (Cb) -- (Cc) -- (Ca);
	\draw[dashed] (0,2.5,0) -- (5,2.5,0) -- (3,2.5,5) -- (0,2.5,0);
	\draw[dashed] (0,2.5,0) -- (0,0,0);
	\draw[dashed] (5,2.5,0) -- (5,0,0);
	\draw[dashed] (3,2.5,5) -- (3,0,5);
	
	\draw[ultra thick] (Ca) -- (Qa);
	\draw[ultra thick] (Cb) -- (Qb);
	\draw[ultra thick] (Cc) -- (Qc);
	
	\draw[-triangle 45][red] (Ca) -- (Qc);
	\draw[-triangle 45][red] (Cb) -- (Qa);
	\draw[-triangle 45][red] (Cc) -- (Qb);
	
	\draw[fill=blue] (Qa) circle (0.15);
	\draw[fill=blue]  (Qb) circle (0.15);
	\draw[fill=blue]  (Qc) circle (0.15);
	
	\node [above right] at (Qa) {$r_0$};
	\node [above left] at (Qb) {$r_1$};
	\node [above right] at (Qc) {$r_2$};
	
	\draw[fill=black] (Ca) circle (0.15);
	\draw[fill=black] (Cb) circle (0.15);
	\draw[fill=black] (Cc) circle (0.15);
	
	\node [left] at (Ca) {$c_0$};
	\node [below right] at (Cb) {$c_1$};
	\node [below right] at (Cc) {$c_2$};
	
	\draw[fill=yellow] (PS) circle (0.15);
	\node [right] at (PS) {$s$};
    
    \begin{scope}[shift={($(0,-3.5,0)$)}]       
    \draw[->] (-3,0,1) -- (-3,0,-0.2);
    \node [right] at (-3,0,-0.2) {$y$};
    \draw[->] (-3,0,1) -- (-2,0,1);
    \node [right] at (-2,0,1) {$x$};
    \draw[->] (-3,0,1) -- (-3,1,1);
    \node [right] at (-3,1,1) {$t$};
    \end{scope}
	
\end{tikzpicture}
\end{center}
\caption{The three diamond task described in text. The known protocol for completing this task makes use of quantum error-correction: The system is encoded into a $((2,3))$ secret sharing scheme with one share sent to each of the call points $c_i$. The shares are then routed to $r_{i+1 \text{ mod } 3}$ if $b_i=0$, and to $r_i$ if $b_i=1$. This task is the simplest example of a summoning task which Alice can complete if there is a guarantee Bob will make only one call, but not if Bob may make an arbitrary number of calls.}\label{fig:triangles}
\end{figure}
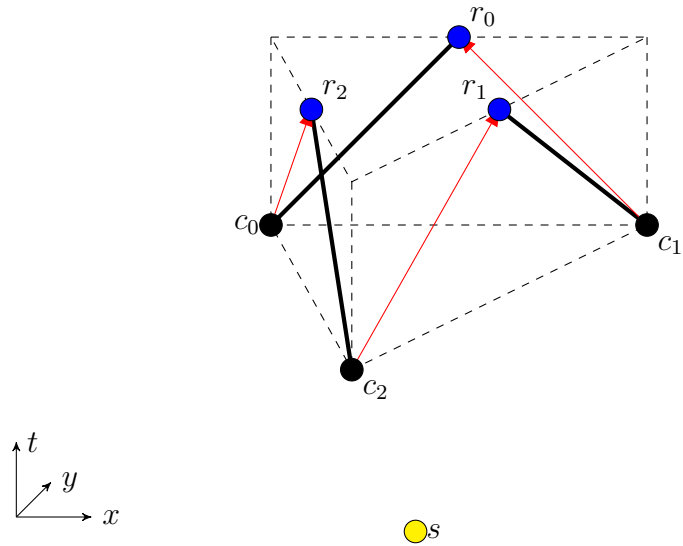

The simplest case where the conditions of many-call single-return summoning and those for many-call many-return summoning differ is the three diamond task shown in figure \ref{fig:triangles}. Consider the arrangement of diamonds shown there, and take any set of diamonds to be authorized. Then to complete the many-call many-return task Alice encodes the system $A$ into a $((2,3))$ secret sharing scheme and sends one share to each of the call points $c_i$. She then routes each share according to the bits $b_i$ she receives at each point; if $b_i=0$ she forwards the share to the next return point $r_{i+1}$, while if $b_i=1$ she sends the share to the return point $r_i$. One can readily check that if one or two calls are sent two shares will end up at a single return point, and the system is handed over at a single diamond. 

However, if a call is sent to all three diamonds, only one share ends up at each diamond. Indeed, Adlam and Kent showed that the unrestricted-call single-return task is impossible on this three diamond arrangement. This is interesting, but we argue it does not indicate that the system cannot be localized to each diamond, at least using the notion of localized we employ in this article. In the protocol using the $((2,3))$ secret sharing scheme, two shares pass through each diamond when Bob sends no calls. Someone with full access to the region enclosed by any one diamond can gather both these shares from the secret sharing scheme and later use them to construct the system. Thus, in this sense the system is localized to all three diamonds.

When Bob sends a call, however, he may prevent the system from being reproduced in certain diamonds. This is obvious in a more prosaic example: Suppose we have two diamonds, with a diamond $D_2$ far in the causal future of the diamond $D_1$. Then Bob giving a call to $D_1$ results in Alice handing the system over to Bob there, and so she does not produce the system in diamond $D_2$. One thing that is interesting about the three diamond task, as revealed by Adlam and Kent, is that in some cases Bob's calls can prevent the system from being reproduced in any diamond. In particular this can happen in cases with cyclic connections among diamonds, as in the three diamond task. 

\bibliographystyle{unsrtnat}
\bibliography{biblio}

\end{document}